\documentclass{extarticle}
\usepackage{geometry}
\usepackage{times}
\usepackage{soul}
\usepackage{url}
\usepackage[hidelinks]{hyperref}
\usepackage[utf8]{inputenc}
\usepackage[small]{caption}
\usepackage{graphicx}
\usepackage{amsmath}
\usepackage{amsthm}
\usepackage{booktabs}
\usepackage{algorithm}
\usepackage{algorithmic}
\usepackage[switch]{lineno}

\usepackage[T1]{fontenc}

\newtheorem{example}{Example}
\newtheorem{theorem}{Theorem}
\newtheorem{proposition}{Proposition}

\newtheorem{remark}{Remark}
\newtheorem{definition}{Definition}

\usepackage{natbib}
\usepackage{cleveref}
\crefname{figure}{figure}{figures}
\hypersetup{
		pdfencoding=auto, 
		psdextra,
		colorlinks=true,
		citecolor=green!40!black,
		linkcolor=red!50!black,
		urlcolor=blue!80!black
	}
\usepackage{subcaption}

\usepackage{nicefrac}
\usepackage{amssymb}

\usepackage{fontawesome5}
\usepackage{pifont}
\usepackage{xcolor}

\let\oldemph\emph
\renewcommand{\emph}[1]{\textcolor{blue!60!black}{\oldemph{#1}}}

\usepackage{tikz}
\usepackage{pgfplots}
\usepgfplotslibrary{groupplots}
\pgfplotsset{compat=1.17,
	legend image code/.code={
		\draw[mark repeat=2,mark phase=2]
		plot coordinates {
			(0cm,0cm)
			(0.15cm,0cm)        %
			(0.3cm,0cm)         %
		};%
}}

\DeclareMathOperator*{\argmax}{argmax}
\DeclareMathOperator{\car}{\#}
\DeclareMathOperator{\conf}{conf}
\DeclareMathOperator{\nconf}{nonconf}

\newcommand{\yes}{\textcolor{mygreen}{\ding{51}}}
\newcommand{\no}{\textcolor{myred}{\ding{55}}}

\newcommand{\Vab}{V^{a\succ b}}
\newcommand{\Vba}{V^{b\succ a}}
\newcommand{\Vxy}{V^{x\succ y}}
\newcommand{\Vyx}{V^{y\succ x}}

\newcommand{\R}{\mathcal{R}}

\newcommand{\fight}{\faPeopleArrows}

\usepackage{appendix}
\usepackage{etoolbox}

\newcommand{\appendixsection}[2]{%
 \gappto{\appendixProofs}{
   \section{Additional Material for Section~\ref{#1}}\label{app:#1}
	 #2
 }
}

\newcommand{\appendixsubsection}[2]{%
 \gappto{\appendixProofs}{
   \subsection{#1}
	 #2
 }
}
\newcommand{\appendixproof}[3]{%
  \gappto{\appendixProofs}
  {
		\subsection{Proof of~{#1}~\ref{#2}}\label{proof:#2}
    #3
  }
}

\title{Selecting the Most Conflicting Pair of Candidates}
\author {
    Théo Delemazure,\textsuperscript{\rm 1}
    Łukasz Janeczko,\textsuperscript{\rm 2}\\
    Andrzej Kaczmarczyk,\textsuperscript{\rm 2}
    Stanisław Szufa,\textsuperscript{\rm 1,2}\vspace{0.3cm}\\
    \textsuperscript{\rm 1}CNRS, LAMSADE, Université Paris Dauphine - PSL,
		Paris, France\\
    \textsuperscript{\rm 2} AGH University, Kraków, Poland\\
    {\small
    theo.delemazure@dauphine.eu,
  	ljaneczk@agh.edu.pl,}\\
		{\small
    andrzej.kaczmarczyk@agh.edu.pl,
    s.szufa@gmail.com}
}
\date{}

\definecolor{myred}{RGB}{201, 22, 22}
\definecolor{mygreen}{RGB}{38, 150, 68}

\begin{document} 

\maketitle

\begin{abstract}
	We study committee elections from a perspective of finding the most
	conflicting candidates, that is, candidates that imply the largest amount of
	conflict, as per voter preferences. By proposing basic axioms to capture this
	objective, we show that none of the prominent multiwinner voting rules meet
	them. Consequently, we design committee voting rules compliant with our desiderata,
	introducing \emph{conflictual voting rules}. A subsequent deepened analysis
	sheds more light on how they operate. Our investigation identifies various
	aspects of conflict, for which we come up with relevant axioms and
	quantitative measures, which may be of independent interest. We support our
	theoretical study with experiments on both real-life and synthetic data.
\end{abstract}

\section{Introduction}

Where a collective decision over a set of options based on a number of opinions
has to be reached,
conflict is inevitable.
Reflected by differences in the opinions, it
usually comes from different perspectives 
of the opinions
(e.g.,\ in the case of human opinions, a beginner investor would likely have
completely different opinion on various asset classes than their professional
counterpart). However, conflict
might also be option-based and
stem from diverse, sometimes even contradicting, inherent qualities
of the options (e.g.,\ potentially high-return assets typically have high risk
levels).

We are interested in how to identify these conflicting options, based on the
preferences. Since the options might represent multiple entities (e.g.,\
sports players, societal issues, or marketing strategies), answering this question has
numerous natural applications that include selecting competitors to organize
engaging sport events (e.g.,\ boxing matches), controversial topics to organize
interesting political debates, disputable issues for socially-relevant
deliberations, or conflicting ideas for boosting the creativity with passionate discussions.

A somewhat different application of identifying conflicting options is learning
new insights about the options or the opinions' perspectives. Imagine a space
agency that validates procedures (options) for landing on the Moon using a
collection of complex simulations. Each simulation assesses the quality of each
procedure and ranks the procedures (expressing an opinion) from the ones that
are most likely to the one that are the least likely to succeed. The existence of two
significantly conflicting procedures can then offer additional insights. It
might suggest that there is some (possibly unknown) feature that impact only
some of the simulations and is ignored by the rest. Alternatively, the two
procedures may differ in some operational detail that is a crucial success factor
for some of the simulated scenarios. In both cases, a careful inspection of the
procedures would help to recognize the source of the conflict and thus
contribute to advancing the explainability of the simulations or the knowledge
about the procedures.

To provide a big picture of our approach, we use a particularly illustrative
application, which is finding polarizing issues.  
Having a collection of
\emph{ordinal preferences} of \emph{voters} expressing their view on the
importance of the issues, we aim at identifying \emph{two} issues that are the
\emph{most conflicting} ones.  This goal poses a significant conceptual
challenge as conflict, which in our scenario can be associated with polarization,
is actually of dual nature. Indeed, we want to find two issues that are (1)
supported by two different, large,
and, ideally, equal-sized groups of voters and that are (2) perceived
ideologically as far from each other as possible in these two groups. As it turns
out, balancing between these two conditions forms an important part of our
modeling of conflict.
A natural motivation behind our goal in this scenario is to reduce the
polarization. This motivation combines
the two aforementioned applications: of selecting two disputable issues and of
fostering learning about the selected issues and the voters. Indeed, learning
about the groups of voters can help in finding communication channels to improve
dialogue between the groups. The conflicting issues, on the other hand, are the
first ones to be addressed by specific policies. Additionally, the selected
issues might indicate particular compromises that must be obeyed not to divide
the society more.

We view our task of finding conflicting candidates naturally falling into the
framework of \emph{multiwinner voting rules}. Such rules, given
\emph{preferences} of \emph{voters} over \emph{candidates} and a desired
committee size, select a winning \emph{committee}\,---\,a subset of the
candidates\,---\,of the desired size that aims at meeting certain objectives.
Adapting this terminology to our scenario, we also have a number of candidates 
and a collection of voters\footnote{For convenience and
increased readability, we decided to keep term ``voters,'' even though
,as demonstrated in the introduction,
our model and its applications are not limited to voting scenarios.}
expressing preferences over the candidates. Our objective
is to select the most conflictual candidates.

In their survey, \citet{fal-sko-sli-tal:b:multiwinner-voting} discuss three
so-far studied objectives of multiwinner voting rules: (1) individual
excellence, focusing on selecting individually the best candidates, (2)
diversity, which aims at representing as many voters as possible, and (3)
proportional representation, which selects a committee that proportionally
represents the voters.
There seems to be a clear consensus that the two former goals are achieved by,
respectively, the $k$-Borda~\citep{deb:j:k-borda} and the
Chamberlin-Courant~\citep{cha-cou:j:cc,elk-fal-sko-sli:j:multiwinner-properties}
rules.
Understanding how various rules 
guarantee achieving proportionality is an active area of
research~\citep{mon:j:fully-proportional-representation,elk-fal-las-sko-sli-tal:c:multiwinner-voting,elk-fal-sko-sli:j:multiwinner-properties,sko-lac-bri-pet-elk:c:proportional-rankings,fal-sko-szu-tal:c:stv-pav,per-pie-sko:c:proportional-pb}.

Focusing on diverse candidates might appear as being able to fulfil our goal of
seeking
conflictual candidates. Intuitively, two most diverse candidates should be
exactly those that divide the voters the most. However, 
as we show in \Cref{sec:axiom-standard}, a fairly small example shows
that this is not the case (in fact, the example shows that none of the standard three
objectives meets our needs). We also verify this considering scenarios beyond
the ``worst-case'' ones in our simulations.

To the best of our knowledge, no rule suitable for finding such conflicting
candidates has been introduced so far. Nonetheless, this line of research
closely aligns with the concept of polarization\footnote{%
	As well as with agreement and diversity of votes, as argued
  by~\citet{fal-kac-sor-szu-was:c:diversity-agreement-polarization}.%
}
in elections extensively studied from
axiomatic, computational, and experimental
perspectives~\citep{alc-vor:j:cohesiveness,has-end:c:diversity-indices,alc-vor:j:cohesiveness2,can-ozk-sto:generalized-polarization,col-gra-hid-mac-nav:c:controlling-rank-aggregation,fal-kac-sor-szu-was:c:diversity-agreement-polarization}.
Yet, the focus has so far been on measuring a value of polarization of all
voters or of a single candidate, which is a different problem than
selecting a pair (or a subset) of conflicting candidates.

\begin{table}[t]
\centering

    \begin{tabular}{ l | c |  c | c | c }
     & \rotatebox{90}{MaxSum} & \rotatebox{90}{MaxNash} & \rotatebox{90}{MaxPolar} & \rotatebox{90}{MaxSwap} \\
     \midrule
     Reverse Stability & \yes & \yes & \yes & \yes \\
     Conflict Consistency & \yes & \yes  & \yes & \yes \\
    \midrule
     Conflict Monotonicity  & \no & \no  & \no & \no \\
     Antagonization Consistency  & \yes & \yes  & \yes &  \yes\\
     Matching Domination & \yes & \yes & \yes &  \no\\
     \midrule 
     Balance Preference & \no & \yes & \no & \yes \\
    \bottomrule
    \end{tabular}
\caption{Axiomatic properties of conflictual rules.}
\label{tab:properties}
\end{table}

\paragraph{Our Contributions.} 
Our primary contribution is conceptual, offering a novel perspective on the
voting theory inspired by finding conflicting candidates. To this end, we
establish new foundational axioms and show that they differ from the
traditional ones. Defining more demanding axioms, we arrive at an impossibility
result that sets our expectations (\Cref{sec:properties}). This result also
guides us in developing new multiwinner voting rules, termed \emph{conflictual
voting rules}, in~\Cref{sec:rules}. They are specifically designed to identify
the most conflicting pair(s) of candidates as required by the introduced axioms (\Cref{tab:properties} highlights our axiomatic analysis). To better understand the intricacies of
the introduced rules, in~\Cref{sec:interpretation}, we provide multiple metrics
capturing different aspects of conflict, which can be of independent interest.

Finally, in~\Cref{sec:experiments}, we present the experimental evaluation, using
both synthetic and real-life elections, empirically showing the behavior of the
conflictual rules.\footnote{The source
code for our experiments is open-source and publically available
at~\url{https://github.com/Project-PRAGMA/conflictual-rules--IJCAI-24}.} Thus, we not only validate our theoretical insights but also
provide insights into their real-world applications and implications.

\section{Preliminaries} \label{sec:prelim}

For a set $X$, let $\car{}X$ denote the cardinality of that set.
Let~$E=(C,V)$ be an election with a set~$C=\{c_1,\dots,c_m\}$ of candidates and
a set~$V=\{v_1,\dots,v_n\}$ of voters. A profile~$P = (\succ_1, \dots,
\succ_n)$ is a collection of rankings (total linear orders) over $C$ such that
each voter $v_i$ is associated with~$\succ_i$; we denote the set of all possible
profiles by~$\mathcal{P}$. By $\overleftarrow{P} = (\overleftarrow{\succ}_1,
\dots, \overleftarrow{\succ}_n)$, we denote profile $P$ with all ballots being
reversed. In particular, for all $v_i \in V$ and all $c,c' \in C$, $c \succ_i
c'$ if and only if $c' \overleftarrow{\succ}_i c$. A committee voting rule $\R$
is a function that takes as input a profile $ P \in \mathcal P$ and a committee
size $k \ge 2$ and outputs a non-empty set~$\R(P,k)$ of committees, such that
for each $W \in \R(P,k)$ we have~$\car{}W = k$ and~$W \subseteq C$. In this document, we focus on the case $k = 2$, and write $\R(P)$ for simplicity.

For brevity, we denote by $v(a)$ the position of candidate~$a$ in vote~$v$. More formally, we have $v_i(a) = \car \{x \mid x \succ_i a \}+1$. Moreover, let $v(a, b)$ denote the distance between $a$ and $b$ in vote $v$, that is, $v(a, b) = v(b) - v(a)$; for conciseness instead of~$v(a,b)$, we
write~$v(ab)$. For instance, in the vote $v = a \succ b \succ c \succ d \succ e$, we have $v(b) = 2$, $v(e) = 5$ and $v(be) = v(b,e) = 5 - 2 = 3$. Analogously $v(eb) = -3$.

Furthermore, for all pairs~$\{a, b\}$ of candidates from~$C$, we denote
$\Vab = \{v_i \in V \mid a \succ_i b \}$ the set of voters preferring $a$ to $b$
in profile $P$. Finally, we say that a pair of candidates $\{a,b\}$ is conflicting
if their ordering is not unanimous, i.e., $\car \Vab > 0$ and $\car \Vba >
0$. In other words, a pair of candidates is conflicting if neither of them 
Pareto-dominates the other one.

\section{Properties of Conflictual Rules} \label{sec:properties}

Before defining our rules, we proceed with fundamental properties that we
require from them. As we will show, already these somewhat weak, basic
axioms prove that the objectives of conflictual rules are far from those of the
rules commonly studied in the computational social choice literature.

\subsection{Fundamental Axioms}

We start by defining two fundamental axioms expected from conflictual rules. The first one assures that unless we have an identity election, only a conflicting pair could win.

\begin{definition}[Conflict Consistency]
	Rule $\R$ is \emph{conflict consistent} if it does not output non-conflicting 	pair if there exists a conflicting one.
\end{definition}

The second axiom comes from the observation that in conflicting voting rule, we want ``love'' and ``hate'' to have symmetrical purposes. Thus, the most polarizing pair in a profile should be as much polarizing if we reverse all voters' preferences. 

\begin{definition}[Reverse Stability]
	Rule $\R$ is \emph{reverse stable} if for each profile $P$ it holds that
$\R(P) = \R(\overleftarrow{P})$.
\end{definition}

\subsection{Relation to Standard Axioms} \label{sec:axiom-standard}

The postulates of our axioms stand in stark contrast to all established
objectives of multiwinner voting rules, that is, to individual excellence,
proportionality, and diversity.
In particular, conflictual rules are incompatible with the unanimity axiom. This
is a very weak axiom guaranteeing effectiveness, saying that if a candidate is
ranked first by every voter, it should be selected in the committee. This is the
main difference with the classical paradigm of social choice: conflictual rules
specifically {\em avoid} consensual candidates.

\begin{definition}[Unanimity]
	Rule $\R$ is \emph{unanimous} if a candidate that appears on top of every ranking is always in the winning committee.
\end{definition}

\begin{proposition}\label{prop:unanimity-consistency-impossibility}
  There is no rule~$\R$ that satisfies both unanimity and conflict-consistency.
\end{proposition}
\begin{proof}
	Consider the profile $\{a \succ b \succ c, a \succ c \succ b\}$. Unanimity
	implies $a$ belongs to the committee, but the only conflicting pair in this
	profile is $\{b,c\}$.
\end{proof}

The incompatibility stated by~\Cref{prop:unanimity-consistency-impossibility} has
a strong implication. Namely, conflictual rules are different from the very prominent
family of committee scoring rules. The observation follows from the fact that
the latter rules must be unanimous~\citep{sko-fal-sli:j:characterization-csr}.

\subsection{Further Properties of Conflictual Rules}

In this section, we consider more axioms, that sound desirable but are unfortunately not always possible to satisfy.

The first axiom we want to introduce is a kind of Pareto-domination axiom, since it is based on the domination of a pair of candidates by another one. However, the pair needs not to be dominated in every ranking, as in classic Pareto properties.
Given profile~$P$ and two conflicting pairs of candidates~$\{a,b\}$ and $\{x,y\}$ we say that $\{a,b\}$ is \emph{matching-dominating} $\{x,y\}$ if there exists a bijective function $f:V \rightarrow V$ that (1) maps all voters from $\Vab$ to $\Vxy$ and from $\Vba$ to $\Vyx$, 
(2) for all voters $v$, $|v(ab)| \ge |(f(v))(xy)|$ and (3) there exists a voter $v$ such that $|v(ab)| > |(f(v))(xy)|$.

\begin{example}
 Consider profile $\{v_1\colon a \succ x \succ y \succ b,
    v_2\colon a \succ y \succ x \succ b,
    v_3\colon b \succ x \succ a \succ y,
    v_4\colon b \succ y \succ a \succ x\}$ and pairs~$\{a,b\}$ and $\{x,y\}$ of candidates, which gives sets $\Vab = \{v_1, v_2\}$ and $\Vxy = \{v_1, v_3\}$.
    A matching $f =\{v_1 \rightarrow v_1$, $v_2 \rightarrow v_3$, $v_3 \rightarrow v_2$,
    $v_4 \rightarrow v_4\}$ meets condition~(1). Recalling that, e.g.,\ $(f(v_2))(xy)) = v_3(xy) = 2$, it is easy to verify that $|v_i(ab)| \ge |(f(v_i))(xy)|$ for all~$i \in \{1,2,3,4\}$ (condition~(2)). Finally, since $|v_1(ab)| = 3 > |(f(v_1))(xy)| = |v_1(xy)| = 1$, $f$ meets condition~(3).
Hence, the pair $\{a,b\}$ dominates pair $\{x,y\}$.
\end{example}

\begin{definition}[Matching Domination] 
	Rule~$\R$ satisfies \emph{matching domination} if matching-dominated pairs are never selected.
\end{definition}

The next axiom is inspired by the monotonicity notion from the literature of classical social choice. The idea is the following: If one pair is the most conflicting in a given profile, adding more conflict between the two candidates should not make another pair even more conflicting, and be selected instead of the original one. 

\begin{definition}[Conflict Monotonicity]
	Rule~$\R$ is \emph{conflict monotonic} if for each profile $P$ and each selected pair $\{a,b\} \in \R(P)$, it holds that if we increase the distance between $a$
	and $b$ by swapping one of them with its adjacent candidate in one of the votes, $\{a,b\}$ is still selected.
\end{definition}

Unfortunately, we can show that this property is incompatible with conflict consistency and matching domination.

\begin{theorem}
No rule satisfies matching-domination, conflict consistency, and conflict monotonicity.
\end{theorem}

\begin{proof}
    Let $f$ be a rule satisfying these 3 axioms, and consider the following profile: $\{v_1:a \succ b \succ c \succ d,  v_2:b \succ a \succ d \succ c \}$. 
    By conflict consistency, the only pairs that can be selected are $\{a,b\}$ and $\{c,d\}$. Assume that $\{a,b\}$ is selected. The proof if $\{c,d\}$ is selected is almost the same. Now, consider the profile $\{v_1:a \succ b \succ c \succ d, v_2:b \succ d \succ c \succ a\}$, in which we increased the conflict between $a$ and $b$ in the second vote by swapping $a$ with its neighbors. By conflict monotonicity, $\{a,b\}$ should still be selected.
    Now, the values of $v(ab)$ are $(1,-3)$ and the values of $v(ad)$ are $(-3,2)$. With the matching $f = \{v_1 
 \rightarrow v_2, v_2 \rightarrow v_1\}$, we obtain that $\{a,d\}$ is dominating $\{a,b\}$, thus $\{a,b\}$ cannot be selected by matching-domination. This is a contradiction.
\end{proof}

This example highlights why conflict monotonicity is quite hard to achieve, but an extreme (and weaker) version of it can actually be satisfied by conflictual rules. The idea is that instead of increasing the conflict in only one ranking, we increase it in every ranking at once, and we increase it as much as we can in every ranking.

By \emph{antagonization} of profile~$P$ with respect to pair $\{a,b\}$, denoted by~$P^{ab}$, we refer to a profile~$P$ in which, in all votes from $\Vab$ we shift $a$ to the first position and $b$ to the last position,
and in all votes from $\Vba$ we shift $b$ to the first position and
$a$ to the last position. However, the relative order of all the other candidates  remains the same.

\begin{definition}[Antagonization Consistency]
	Rule $\R$ is \emph{antagonization consistent} if, given a profile~$P$
	and a selected pair~$\{a,b\} \in \R(P)$, $\{a,b\}$ is also selected in~$P^{ab}$.
\end{definition}

\section{Conflictual Rules} \label{sec:rules}

After we made explicit what we expect from conflictual rules by the means of
axioms, we can start searching for such rules. All proofs from this section are
available in~\Cref{apdx:proofs}.

In order to define the rules, we should first define what could be a {\em conflict}, in particular between two voters. Intuitively, there is a conflict induced by a pair of candidates $\{a,b\}$ between two voters if they disagree on the ordering between $a$ and $b$. Moreover, the more distant $a$ and $b$ are in the rankings of the voters, the greater the conflict. Starting with this, we define the sum-conflict ($\conf^+$) and the Nash-conflict ($\conf^{\times}$) as follows:

\begin{definition}[Pairwise conflict]
    For $\circ \in \{+, \times\}$, let the conflict induced by a pair of candidates $a$ and $b$ between two votes $v$ and $v'$ be: 
     \[
        \conf^{\circ}_{v,v'}(a,b) = \begin{cases}
    0 & \text{if } v(ab)\cdot v'(ab) > 0 \\
      |v(ab)|\circ|v'(ba)| & \text{otherwise} \\
    \end{cases}
  \]
\end{definition}

For instance, the Nash (resp. sum) pairwise conflict induced by $a$ and $b$ between votes $a \succ b \succ c$ and $b \succ c \succ a$ is $\conf^\times_{v,v'}(a,b) = 1\times2=2$ (resp. $\conf^+_{v,v'}(a,b) = 1+2=3$).
By extension, for two candidates $a$ and $b$, the conflict is defined as the sum of the pairwise conflict over all possible pairs of voters: $\conf^{\circ}(a,b) = \sum_{v,v' \in V} \conf_{v,v'}^{\circ}(a,b)$.

Then, we can define the rules that select the pairs of candidates that maximize this value.
Hence, we define the {\em MaxSumConflict} rule based on $\conf^+$, $\textrm{MaxSum}(P) = \argmax_{a,b \in C} \conf^+(a,b)$,
which is equivalent to
\[
\argmax_{a,b \in C} \quad \car \Vba \sum_{v \in \Vab} v(ab) +
\car \Vab \sum_{v \in \Vba} v(ba).
\]

Similarly, we define the {\em MaxNashConflict} rule, $\textrm{MaxNash}(P) = \argmax_{a,b \in C} \conf^{\times}(a,b)$,
which is equivalent to
\[
	\argmax_{a,b \in C} \sum_{v \in \Vab} v(ab) \cdot
 \sum_{v \in \Vba} v(ba).
\]

\begin{remark}
Looking at the second rule, it is tempting to define an alternative rule in
which we maximize the sum $ \sum_{v \in \Vab} v(ab) + \sum_{v \in \Vba} v(ba)$
instead of the product. However, such a rule does not satisfy
conflict-consistency, and could elect a non-conflicting pair.%
\footnote{Selecting only from conflicting pairs makes this rule
conflict-consistent, yet somehow non-monotonic. Indeed, consider a
non-conflicting pair~$A$ that scores higher than a conflicting pair~$B$. After
adding a voter equally increasing the pairs' scores and making~$A$ conflicting,
the rule selects~$A$, which is intuitively ``less conflicting.'' So, we do not
consider rules naturally failing conflict-consistency.}
\end{remark}

The following example demonstrates the defined rules.

\begin{example}\label{ex:rules-demo}
Let profile~$P$ over $6$~candidates be:
\begin{align*}
    a \succ x \succ c \succ d \succ y \succ b,\\
    c \succ y \succ b \succ a \succ x \succ d.
\end{align*}
The only conflicting pairs of candidates are~$\{a,b\}$ and~$\{x, y\}$. We have
that
$\conf^{+}(a,b) = 6 = \conf^{+}(x,y)$, 
$\conf^{\times}(a,b) = 5$, and~$\conf^{\times}(x,y) = 3 \times 3 = 9$.
Hence, while in the MaxSum rule both pairs tie, MaxNash clearly selects~$\{x,y\}$.
\end{example}

\begin{proposition}\label{properties:maxsum-maxnash}
    MaxSum and MaxNash satisfy reverse stability, conflict consistency, antagonization consistency, and matching-domination.
\end{proposition}
\appendixproof{Proposition}{properties:maxsum-maxnash}{%
\begin{proof}
We independently prove each axiom, one by one.

		\paragraph{Reverse Stability.} Both rules score each pair of
		candidates. The computed score only depends on the values of $v(ab)$ and
		$v(ba)$, and is always symmetric with respect to $a$ and $b$. Thus,
		reversing the profile will lead to the same score for each pair of
		candidates.

    \paragraph{Conflict Consistency.} 
    Both rules associate a score to each pair of candidates. If the candidates are non-conflicting, the score is always $0$, but if the pair is conflicting, then the score is strictly positive. Thus, a non-conflicting pair cannot be selected if there is at least one conflicting pair.

    \paragraph{Antagonization Consistency.}
    Let us show that rules MaxSum and MaxNash satisfy this axiom. Assume
		$\{a,b\}$ is selected in a profile and we antagonize them (we push them
		towards their corresponding extremes). It is clear that no other pair
		involving $a$ or $b$ will get a better score. Indeed, consider for instance
		pair $\{a,c\}$ with $c \ne b$. Since $a$ is always first or last in every
		ranking, we have $\Vab = V^{a \succ c}$ and $\Vba = V^{c \succ a}$ and for
		all voters $v$, $|v(ac)| < |v(ab)|$. This is enough to prove that $\{a,c\}$
		gets a strictly lower score than $\{a,b\}$. The same argument works for
		pairs involving $b$ and not $a$. Now, the score of any other pair of
		candidates $\{x,y\}$ can only decrease when we antagonize $a$ and $b$, while
		the score of $\{a,b\}$ can only increase. Indeed, the sets $\Vxy$ and $\Vyx$
		remain identical, and as we push $a$ and $b$ towards the extreme, we can
		only decrease the distance between $x$ and $y$. Thus, there is no way that
		$\{x,y\}$ has a strictly better score than $\{a,b\}$ after the
		antagonization but not before. This concludes the proof that $\{a,b\}$ will
		still be selected after antagonization.

    \paragraph{Matching-Domination.}
		By their definitions, rules MaxSum and MaxNash satisfy this axiom, because
		any matching-dominated pair $\{x,y\}$ would have a strictly lower score than
		any pair $\{a,b\}$ that matching-dominates $\{x,y\}$.
    
    Let us show it formally. Let $f: V \rightarrow V$ be a bijective function satisfying matching domination definition. As $f$ maps elements from $V^{a \succ b}$ and $V^{x \succ y}$ bijectively, $\quad \car \Vab = \quad \car \Vxy$, analogously $\quad \car \Vba = \quad \car \Vyx$. 
    Further, by summing up the inequalities from the definition we get that
    \begin{align*} 
    \sum_{v \in V^{a \succ b}}{v(ab)} = &\sum_{v \in V^{a \succ b}}{|v(ab)|} \\
    \geq& \sum_{v \in V^{a \succ b}}{|f(v)(xy)|} \\
    = &\sum_{v \in V^{x \succ y}}{|v(xy)|} = \sum_{v \in V^{x \succ y}}{v(xy)}
    \end{align*}
    and 
    \begin{align*}
        \sum_{v \in V^{b \succ a}}{v(ba)} &= \sum_{v \in V^{b \succ a}}{|v(ab)|} \\ &\geq \sum_{v \in V^{b \succ a}}{|f(v)(xy)|} \\
        &= \sum_{v \in V^{y \succ x}}{|v(xy)|} = \sum_{v \in V^{y \succ x}}{v(yx)}.
    \end{align*} 
    
		Consequently, $\sum_{v \in V^{a \succ b}}{v(ab)} \geq \sum_{v \in V^{x \succ
		y}}{v(xy)}$ and $\sum_{v \in V^{b \succ a}}{v(ba)} \geq \sum_{v \in V^{y
		\succ x}}{v(yx)}$, whereas at least one of these inequalities must be strict
		because there exists a voter $v \in V: |v(ab)| > |f(v)(xy)|$.
    
    Taking all into account, we see that 
    \begin{align*}
        &\quad \car \Vba \sum_{v \in \Vab} v(ab) + \car \Vab \sum_{v \in \Vba} v(ba) \\
        >& \quad \car \Vyx \sum_{v \in \Vxy} v(xy) + \car \Vxy \sum_{v \in \Vyx} v(yx)
    \end{align*}
    and
    \begin{align*}
        \sum_{v \in \Vab} v(ab) \cdot \sum_{v \in \Vba} v(ba) >  \sum_{v \in \Vxy} v(xy) \cdot \sum_{v \in \Vyx} v(yx),
    \end{align*}
    so MaxSum and MaxNash will never select a matching-dominated pair of candidates.
    \end{proof}

}

Another approach is to select pairs of candidates $\{a, b\}$  that maximize the
minimum number~$\nconf(a,b)$ of swaps of adjacent candidates to make~$\{a, b\}$
a non-conflicting pair in the profile $P$.
We call the respective rule
{\em MaxSwap} and remark that formally MaxSwap selects pairs $\{a,b\}$ such that:
\[
\argmax_{a,b \in C} \min\left(\sum_{v \in \Vab} v(ab), 
 \sum_{v \in \Vba} v(ba)\right).
\]
For demonstration recall~\Cref{ex:rules-demo}, where MaxSwap would select~$\{x,y\}$
as we need~$3$~swaps to make this pair non-conflicting, compared to one swap
required for~$\{a,b\}$.

\begin{proposition}\label{properties:maxswap}
	MaxSwap rule satisfies reverse-stability, conflict consistency, and
	antagonization consistency. MaxSwap fails matching-domination.
\end{proposition}
\appendixproof{Proposition}{properties:maxswap}{%
\begin{proof}

The arguments to prove that MaxSwap satisfies \textbf{Reverse Stability}, \textbf{Conflict Consistency} and \textbf{Antagonization Consistency} are exactly the same as those for MaxSum and MaxNash, hence, we refer to the proof of \Cref{properties:maxsum-maxnash} for these axioms. Now, we provide a counter-example to show that MaxSwap does not satisfy \textbf{Matching-Domination}.

    \paragraph{Matching-Domination.}
		Let us look at the following profile: $P = \{10 \times d \succ b \succ a
		\succ c, 1  \times a \succ c \succ b \succ d, 1  \times c \succ a \succ d
	\succ b\}$. Pairs $\{a,b\}$ and $\{c,d\}$ both need at least $4$ swaps to
	become ranked by all voters in the same order. It can be checked that all
	other pairs of candidates can be made non-conflicting with fewer than (or
	exactly) $4$ swaps. Thus, MaxSwap would select both pairs $\{a,b\}$ and
	$\{c,d\}$ (in fact, it would also select $\{a,d\}$, $\{b,c\}$, but we do not
	focus on them). However, the vector of the $v(ab)$ values is equal to $(10
	\times -1,2,2)$ and the vector of the $v(cd)$ values is equal to $(10 \times
	-3,2,2)$. From these vectors, we clearly see that $\{c,d\}$ dominates
	$\{a,b\}$ with the identity matching $f(v) 
=v$ for all voters $v$. By matching domination, $\{a,b\}$ should not be selected, but it is. Thus, MaxSwap fails this axiom.
\end{proof}%

}

\section{Interpretation} \label{sec:interpretation}
So far, we have introduced several conflictual rules and provided their
axiomatic analysis. We now look further, beyond the somehow worst-case study
that axioms usually offer. We want to understand what are the practical
differences between our rules. To this end, we introduce various notions that
help us interpret the rules' behavior.

\subsection{Partitioning and Discrepancy}\label{sec:antipathy-and-partitioning}

Intuitively, a pair of candidates is perfectly polarizing if there are two groups of equal sizes that have conflicting preferences on this pair, and all voters have very strong opinions of the candidates. 
We adapt these two features to describe our committees in the notions of social \emph{partitioning ratio} and candidates' \emph{discrepancy}. %
On the one hand, candidates $a$ and $b$ provide maximal partitioning ratio if
exactly half of the voters prefer $a$ to $b$. Formally, we define partitioning
ratio~$\alpha(a,b) \in [0,1]$ for candidates~$a$ and~$b$ as
\[
\alpha(a,b) = \frac{2}{n} \min(\car \Vab, \car \Vba).
\]
On the other hand, candidates $a$ and $b$ %
have high discrepancy if
every voter strongly prefers one candidate to the other. In particular, if each voter ranks
either $a$ or $b$ first, and the other candidate last, then the pair $\{a,b\}$ has maximal discrepancy.
Formally, we define the discrepancy~$\beta(a,b) \in [0,1]$ as

\[
\beta(a,b) = \frac{1}{n\cdot(m-1)}\sum_{v \in V} |v(ab)|.
\]

Given a profile $P$, by $\alpha_{\textrm{max}}(P)$ and
$\beta_{\textrm{max}}(P)$ we denote the $\max_{a,b} \alpha(a,b)$ and $\max_{a,b}
\beta(a,b)$, respectively.
First note that $\beta_{\textrm{max}} \in (\nicefrac{1}{3}, 1]$ since there is always a pair of candidates $\{a,b\}$ such that $\beta(a,b) > \nicefrac{1}{3}$. The complete calculations are available in~\Cref{app:sec:antipathy-and-partitioning}, in which we also provide the $\alpha$ and $\beta$ values of characteristic profiles studied by~\citet{fal-kac-sor-szu-was:c:diversity-agreement-polarization} (i.e., identity, uniformity, and antagonism). 

When selecting conflicting pairs, rules must do a trade-off between these two notions.
To build our intuition, we look at discrepancy and partitioning ratio in a
specific case. %
With the new insights we obtain, we develop a family of voting rules based on $\alpha$ and $\beta$.
Let us fix some candidates~$a$ and~$b$ and consider the case in which the value $|v(ab)|$ is the same for all voters $v \in V$; hence, we have
$\beta(a,b) = \nicefrac{|v(ab)|}{(m-1)}$ (for ease of presentation, we omit $a$
and~$b$ and use $\alpha$ and $\beta$ in this paragraph). Then, using some
constant $C_1$ independent of~$a$ and~$b$,
we can express
the sum-conflict between $a$ and~$b$ as follows:
\begin{align*}
		\conf^+(a,b)&= \car \Vba \sum_{v \in \Vab}\!\!v(ab) + \car \Vab \sum_{v \in
		\Vba}\!\!v(ba)  \\
    &= 2\car \Vba\car \Vab \beta \\
    &= 2\left(n\frac{\alpha}{2}\right)\left(n\left(1-\frac{\alpha}{2}\right)\right)\beta = C_1\alpha(2-\alpha)\beta.
\end{align*}
Analogously,
$\conf^\times(a,b) = C_2\alpha(2-\alpha)\beta^2$ and $\nconf(a,b) = C_3\alpha\beta$
(see~\Cref{app:sec:antipathy-and-partitioning} for the derivations).
The
expressions clearly illustrate the tension between discrepancy and partitioning
ratio: while some rules give more weight to pairs that divide the society more equally, other
prefer those that have higher discrepancy.

\appendixsection{sec:antipathy-and-partitioning}{}
\appendixsubsection{Expressing MaxNash and MaxSwap Using Group Discrepancy Imbalance}{
For these two rules, the groups $\Vab$ and $\Vba$ have symmetric roles. Let us assume w.l.o.g. that $\car \Vab \ge \car \Vba$. Then, we have $\car \Vba = n\nicefrac{\alpha}{2}$ and $\car \Vab = n(1-\nicefrac{\alpha}{2})$.
For MaxNash, the score of $\{a,b\}$ would be
 \begin{align*}
     S(a,b) &= \sum_{v \in \Vab}v(ab) \cdot \sum_{v \in \Vba}v(ba)  \\
     &= \car \Vab(n(m-1)\beta)\car \Vba(n(m-1)\beta)\\
     &= (n(m-1))^2n^2\nicefrac{\alpha}{2}(1-\nicefrac{\alpha}{2})\beta^2 \\ 
     &=C\alpha(2-\alpha)\beta, 
 \end{align*}
where $C$ is a constant independent of $\{a,b\}$.

For MaxSwap, the score of a pair of candidates $\{a,b\}$ would be 
 \begin{align*}
     S(a,b) &= \min(\sum_{v \in \Vab}v(ab), \sum_{v \in \Vba}v(ba))  \\
     &= \min(\car \Vab, \car \Vba)n(m-1)\beta \\ 
     &= C\alpha\beta,
 \end{align*}
where $C$ is a constant independent of $\{a,b\}$.
}

The observed trade-off forms in fact a flexible framework for a family of rules
covering the whole spectrum of possible behavior.
Let $\mathcal R$ be some rule endowed with a scoring function~$S(\cdot, \cdot):[0,1]^2 \rightarrow \mathbb R_{\ge 0}$
that selects pairs of candidates~$\{a,b\}$ which maximize value~$S(\alpha(a,b),
\beta(a,b))$. Each such rule~$\mathcal R$ naturally satisfies reverse-stability.
Moreover, the rule satisfies conflict consistency if and only if $S(\alpha, \beta) = 0$ when $\alpha = 0$ and $S(\alpha,\beta) > 0$ otherwise. If  $S$ is strictly increasing with both $\alpha$ and $\beta$ when $\alpha > 0$, then the rule based on $S$ also satisfies antagonism-consistency and matching-domination.
In particular, these conditions are met by the family of scoring functions $S(\alpha, \beta) = \alpha\beta^p$ for $p > 0$, that we named \emph{$p$-MaxPolarization} rules
(\emph{$p$-MaxPolar} in short). Note that, the higher the
value of $p$, the more weight is put on discrepancy in comparison to partitioning ratio.
The natural dependence on discrepancy and partitioning ratio
motivates us to include MaxPolar rules into our further analysis.

We visualize how the outcome of our rules depends on discrepancy and
partitioning ratio using the following example. Consider a preference profile
over four candidates~$a_1$, $b_1$, $a_2$, and $b_2$ in which a pair~$\{a_1, b_1\}$ has maximal discrepancy
and a pair~$\{a_2, b_2\}$ has maximal partitioning ratio;
hence, we have
$\beta(a_1,b_1) = 1$ and $\alpha(a_2,b_2) = 1$. Understanding the interplay between
discrepancy and partition ratio boils down to the question: Which pair is preferred in our
scenario for different values of the non-fixed values $\beta(a_2,b_2)$
and~$\alpha(a_1,b_1)$? For clarity, we once again consider our simplified case in
which~$\beta(a_2,b_2) = \nicefrac{|v(a_2,b_2)|}{m-1}$ for all voters $v$. 
\Cref{fig:frontier-rules} depicts how the preferred pair depends on the non-fixed values 
for various rules. For example, for $2$-MaxPolar, the pair~$\{a_1,b_1\}$ is preferred if
$\alpha(a_2,b_2)\beta^2(a_2,b_2) < \alpha(a_1,b_1)\beta^2(a_1,b_1)$, which boils down
to $\alpha(a_2,b_2) < \alpha(a_1,b_1)$ for our simplified scenario. Hence, the figure contains
the plot of~$\alpha(a_1,b_1) = \beta^2(a_2,b_2)$, where pair~$\{a_1,b_1\}$ wins for each point
over the curve. Similarly, for MaxNash, where $S(\alpha, \beta)
=\alpha(2-\alpha)\beta^2$ in our simplified scenario, we see the plot of
$\alpha(2-\alpha) = \beta^2$ (omitting the arguments for~$\alpha$ and~$\beta$ for readability).%

\begin{figure}[t]
    \centering
    \scalebox{.6}{\input{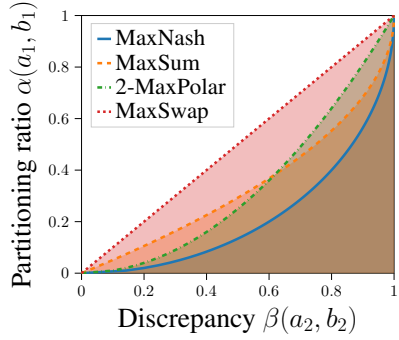}}
    \caption{Area where $\{a_2,b_2\}$ is preferred to $\{a_1,b_1\}$ for different rules. We assume $\alpha(a_2,b_2) = 1$ and $\beta(a_1,b_1) = 1$. }
    \label{fig:frontier-rules}
\end{figure}

\subsection{Polarization Balance} \label{sec:balance}

Partitioning ratio and discrepancy enable us to interpret the rules we
introduced in \Cref{sec:rules} in the special case described above. By this, we
learn whether these rules give more importance to voters having strong
opinions or to be evenly split. Yet, there is another phenomenon distinguishing
these rules from the MaxPolar rules, and that is the discrepancy balance.

To see this, recall profile~$P$ from~\Cref{ex:rules-demo}.
Here, pairs $\{a,b\}$ and $\{x,y\}$ have the same partitioning ratio $\alpha = 1$ and
discrepancy value $\beta = 3/5$. However,
the discrepancy between $a$ and~$b$ is not balanced between the two voters: $a$'s supporter is very extreme in its preferences, while $b$'s supporter is almost indifferent between the two. In comparison, the discrepancy between $x$ and $y$ is perfectly balanced.
MaxPolar rules are agnostic to this phenomenon, while rules from
\Cref{sec:rules} take it into account and would select pair~$\{x,y\}$.

This leads us to introducing our third metric called the \emph{discrepancy balance} $\gamma$. 
By $\mu(a,b) = \frac{\sum_{v\in \Vab} v(ab)}{\car \Vab}$ we denote the average discrepancy between $a$ and $b$ among supporters of $a$. For the discrepancy to be balanced, we want $\mu(a,b)$ and $\mu(b,a)$ to be as similar as possible. Thus, we define it as follows:
\[
\gamma(a,b) =  \min\left(\nicefrac{\mu(a,b)}{\mu(b,a)},\nicefrac{\mu(b,a)}{\mu(a,b)}\right)
\]

Measures $\alpha, \beta$ and $\gamma$ are not independent from each other.
For instance, $\beta(a,b) = 1$ and $\alpha(a,b) > 0$ implies $\gamma(a,b) = 1$
(but $\gamma(a,b) = 1$ does not imply anything for $\beta(a,b)$). Intuitively,
the value of~$\gamma$ has more influence on rules like MaxSwap and MaxNash than
on MaxSum, as the former are more egalitarian than the latter. This intuition is
supported by the experiments presented in \Cref{sec:experiments}.

Note that $\alpha$ measures if the electorate is divided evenly between two candidates
and $\gamma$ measures if the discrepancy is balanced among supporters of each group. By combining these two ideas, we obtain the last measure, called \emph{group discrepancy imbalance} $\phi$, defined as:
\[
\phi(a,b) = \frac{|\sum_{v \in V} v(ab)|}{\sum_{v\in V}|v(ab)|}.
\]

It is~$0$ when the total discrepancies of the groups are equal and~$1$ when they are totally imbalanced (in that case $\alpha = 0$). As expected, contrary to $\gamma$, metric~$\phi$ is sensitive to the size of the groups. Particularly, if all rank differences are equal ($\beta = \nicefrac{|v(ab)|}{m-1}$ for all voters $v$), then $\gamma = 1$, but $\phi = 1 - \alpha$.

With group discrepancy imbalance, we can actually rewrite two of our rules in the general case: MaxNash and MaxSwap. For MaxNash, we have that the score of the pair is proportional to $\beta^2(1-\phi^2)$ and for MaxSwap to $\beta(1-\phi)$, so they both care about this imbalance between the two groups. The other rules cannot be defined using $\phi$ in general, but in the particular case in which all rank differences are the same, we can simply replace $\alpha$ by $1-\phi$, which gives $\beta(1-\phi^2)$ for MaxSum (see~\Cref{app:sec:antipathy-and-partitioning} for calculations).%

\appendixsubsection{Expressing MaxSwap and MaxNash Using Partitioning Ratio}{%
Let us denote
\begin{align*}
    \text{max}_g &= \max(\sum_{v \in \Vab} v(ab),\sum_{v \in \Vba} v(ba)), \\
    \text{min}_g &= \min(\sum_{v \in \Vab} v(ab),\sum_{v \in \Vba} v(ba)).
\end{align*}
We have:
\begin{align*}
    \beta  n(m-1) \phi &= |\sum_{v \in V} v(ab)| &= \text{max}_g - \text{min}_g, \\
    \beta  n(m-1) &= \sum_{v \in V}|v(ab)| &= \text{max}_g + \text{min}_g.
\end{align*}
Thus, we can rewrite the MaxSwap score: 
\begin{align*}
\text{min}_g
&= \frac{\beta  n(m-1) - (\beta   n(m-1)\phi) }{2}\\
    &=\frac{\beta  n(m-1) (1- \phi)}{2} \\
    &= C\beta(1-\phi)
\end{align*}
with $C$ being a constant independent of the pair. This is enough to show that the MaxSwap score is proportional to $\beta(1-\phi)$. Similarly, we can compute MaxNash:
\begin{align*}
&\text{min}_g \cdot \text{max}_g \\
=&\frac{(\beta  n(m-1) - (\beta  \phi  n(m-1)))(\beta   n(m-1) + (\beta  \phi  n(m-1))) }{4}\\
    =&\frac{\beta  n(m-1) (1- \phi)\beta n(m-1)(1+\phi)}{4} \\
    = &C\beta^2(1-\phi^2) 
\end{align*}
with $C$ being a constant independent of the pair.

}
Because of this, we introduce a new property that enable us to distinguish some rules.

\begin{definition}[Balance Preference]
    A rule satisfies \emph{balance preference} if given two pairs $\{a,b\}$ and $\{x,y\}$, if there exists a perfect matching from voters to voters $\phi : V \rightarrow V$ such that $|v(ab)| = |\phi(v)(xy)|$ (i.e., their vectors of absolute rank differences are the same), and if $|\sum_{v\in V} v(ab)| < |\sum_{v \in V} v(xy)|$ (i.e., discrepancy between supporters of $a$ and of $b$ is more balanced), then $\{x,y\}$ cannot be selected.
\end{definition}

To better understand this axiom, consider the profile $P = \{2 \times x{\succ}a{\succ}b{\succ}y,
    1 \times a{\succ}y{\succ}x{\succ}b,
    1 \times b{\succ}y{\succ}x{\succ}a\}$. 
In this profile, voters are evenly divided between $x$ and $y$, but $x$ supporters are extreme while $y$ supporters are quite indifferent. On the other hand, $a$ is preferred to $b$ by 3 voters, but they each have at least one extreme supporter. In that scenario, $\{x,y\}$ is more balanced in terms of group size, but $\{a,b\}$ is more balanced in terms of average polarization. More precisely, $\phi(x,y) = \frac{6-2}{6+2} = \frac{1}{2}$, and $\phi(a,b) = \frac{5-3}{5+3} = \frac{1}{4}$. Moreover, the pairs satisfy the condition of the balance preference axiom, so if a rule satisfies it, $\{x,y\}$ should not be selected. This is why MaxSwap and MaxNash select $\{a,b\}$.  However, $\beta(a,b) = \beta(x,y)$ and $\alpha(a,b) < \alpha(x,y)$, so all MaxPolar rules select $\{x,y\}$, and we can also show that MaxSum selects $\{x,y\}$. %

More generally, we can state the following proposition, which is proven in \Cref{app:sec:antipathy-and-partitioning}.
\begin{proposition}\label{prop:balancepref}
    MaxNash and MaxSwap satisfy balance preference and MaxSum and MaxPolar rules fail it.
\end{proposition}
\appendixproof{Proposition}{prop:balancepref}{%
\begin{proof}
For MaxSwap and MaxNash, we can simply use the fact that they are proportional
to respectively $\beta (1-\phi)$ and $\beta^2 (1-\phi^2)$. Indeed, note that if
there are two pairs $\{a,b\}$ and $\{x,y\}$ such that there exists a perfect
matching $f$ from voters to voters  such that $|v(ab)| = |f(v)(xy)|$ for all
voters, this implies that $\beta(a,b) = \beta(x,y)$ by definition of $\beta$.
This also implies that $\sum_{v \in V}|v(ab)| = \sum_{v \in V}|v(xy)|$. Then, if
$|\sum_{v\in V} v(ab)| < |\sum_{v \in V} v(xy)|$, this implies $\phi(a,b) <
\phi(x,y)$ by definition of $\phi$. Thus, $\beta(a,b)(1-\phi(a,b)) >
\beta(x,y)(1-\phi(x,y))$ so MaxSwap will never select the pair $\{x,y\}$. The
other formula gives the proof for MaxNash.

For MaxSum and MaxPolar, the example provided in \ref{sec:balance} proves that these rules fail this property.
\end{proof}%
}

\section{Experiments} \label{sec:experiments}

We motivate our experiments with two goals: Compare conflictual rules to traditional committee rules, and compare conflictual rules with each other.

\subsection{Comparison with Standard Committee Rules}
First, we compare MaxNash and two well-known committee rules Chamberlin-Courant and Borda, which respectively aim at achieving diversity and individual excellence. We chose MaxNash as a representative of all conflictual rules because all of them gave qualitatively same results. %

Under \emph{Borda} rule, each voter assigns~$m-1$ points to her favorite candidate, $m-2$ points to the second one, $m-3$ to the third one, and so on. Finally, we select the candidates with the highest scores. 

With the \emph{Chamberlin-Courant} rule (CC) and $k=2$, each voter assigns~$m - \min(v(x), v(y))$ points to the pair of candidates $\{x,y\}$, and the pair with the highest score is selected. Thus, only the preferred candidate in the pair is taken into account. 

To get the intuition about differences between these rules and conflictual ones, we use the 2D-Euclidean framework, studied by~\citet{elk-fal-las-sko-sli-tal:c:multiwinner-voting}. In this setting, voters and candidates are associated with their ideal positions in the 2D-Euclidean space, and the preferences of voters are based on their distances to candidates: a voter~$v$ prefers $a$ to $b$ if $v$ is closer to~$a$ than to~$b$.  

\begin{figure}[t]
    \centering
    \includegraphics[width=0.37\textwidth]{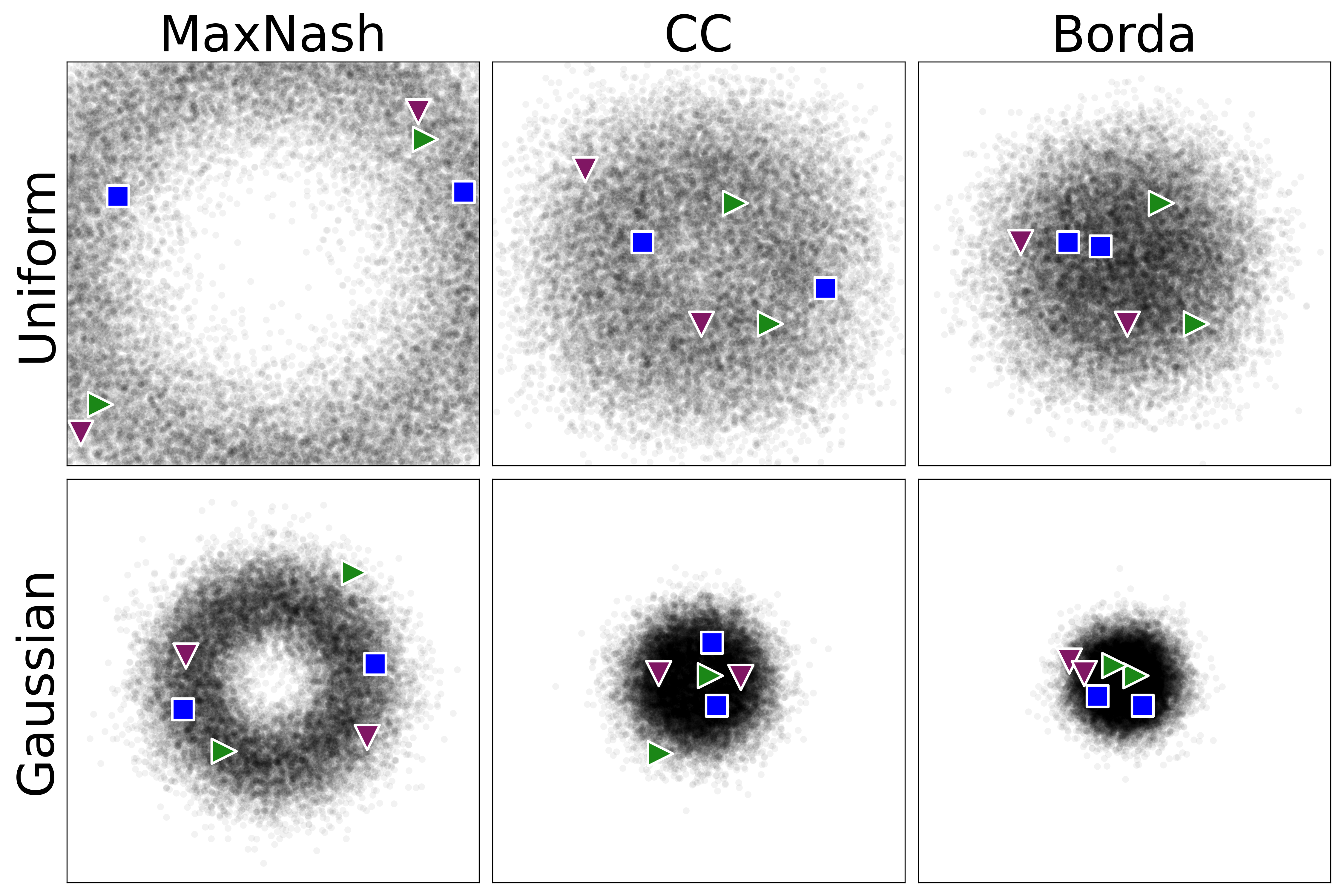}
    \caption{Distribution of the positions of the winning candidates for different rules and distributions of positions. Each pair of colored points correspond to the winners of a single election.}
    \label{fig:2d-pos}
\end{figure}
 
For all voters and candidates, we sample their ideal points from: (1) the uniform distribution in $[0,1]^2$ and (2) the normal distribution centered in $(0.5,0.5)$. In~\Cref{fig:2d-pos} we present the results. We sampled $10,000$ instances, and for each of them marked the two selected candidates on the plane. For three random instances, we marked the winning candidates by colors.

Supporting the axiomatic analysis, the results show that
the conflictual rules behave significantly different from the classical ones. Indeed, they select candidates that are far away from the center of the plane
while all other rules select those that are around 
the center (the closest for the Borda rule, which aims for individual excellence). We also computed partitioning ratio and discrepancy for all pairs of candidates (see \Cref{appsec:expe-results}). Their values confirm that, on average, the pairs selected by Borda and CC have lower values than those selected by the conflictual rules. As expected, this tendency is particularly strong for discrepancy\,---\,in our scenario, if two candidates are close to each other on the plane, then they are relatively close to each other in every ranking.

\subsection{Comparisons of Conflictual Rules} \label{sec:expe-conflicting-rules}

Next, we compare the conflictual rules with each other.
In particular, we focus on MaxSwap, MaxNash, MaxSum, and 2-MaxPolar. %
We ran our experiments on both synthetic and real-life datasets. The synthetic ones include 2D-Euclidean models introduced in the previous subsection, and Mallows model (which are discussed in more depth in Appendix \ref{appsec:expe-mallows}).

For real-life data, we focus on 4 datasets: (i) preferences over $11$ candidates
gathered for experiments during French presidential elections from
2017 and 2022~\citep{voterautrement2018,voterautrement2022}, which we expect to be conflictual, and that we
discuss in more extent in
\Cref{sec:exp-conflictual}, %
similar to a 2-Mallows model, 
(ii) preferences over $10$ sushi types~\citep{sushis} %
and (iii) juries' ranking of contestant performances in figure skating competitions, from Preflib~\citep{preflib}, which we expect to be less conflictual. 
A detailed analysis %
can be found in Appendix \ref{appsec:expe-real}.

\begin{figure}[t]
    \centering
    \includegraphics[width=0.45\textwidth]{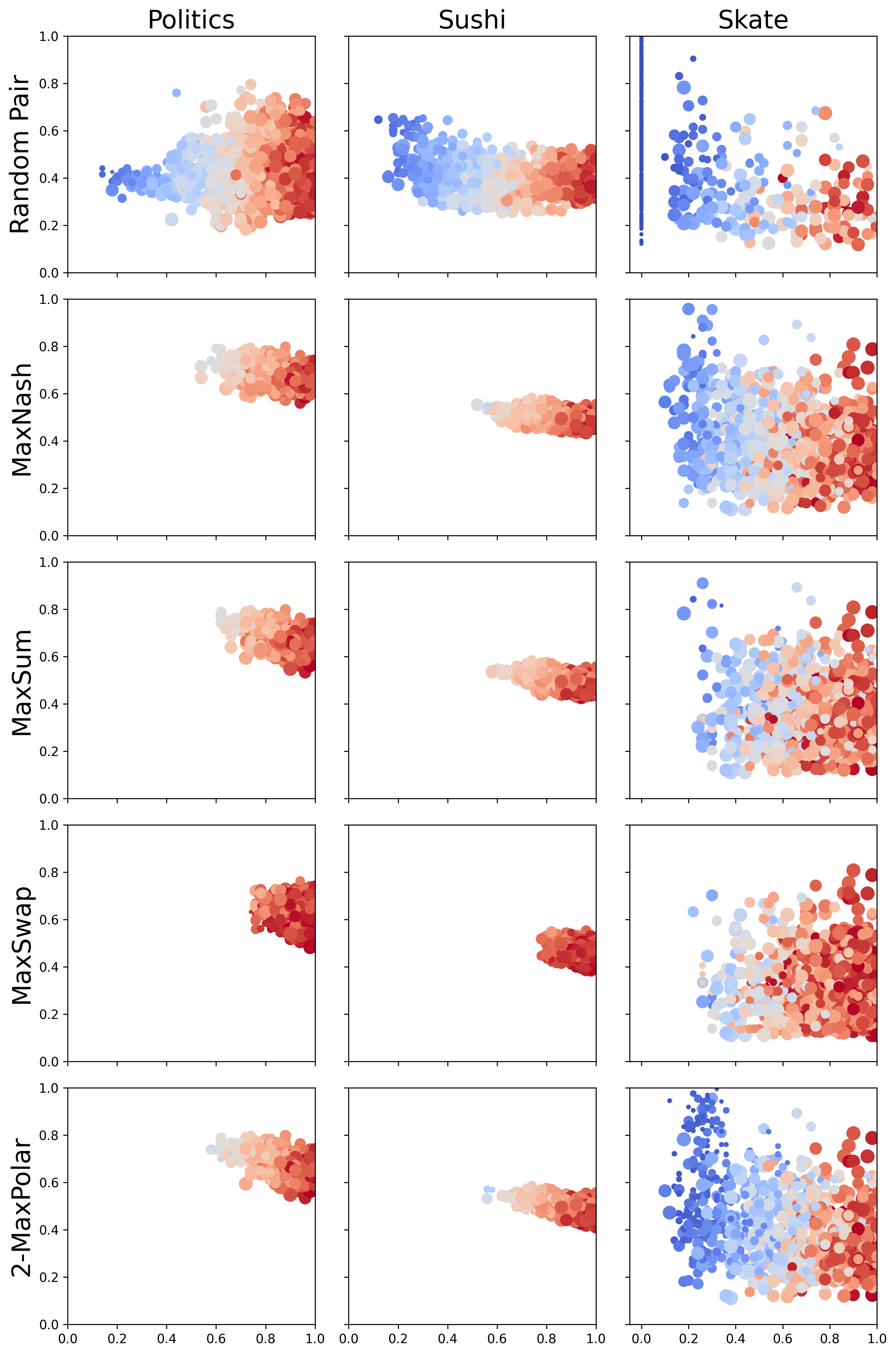}
    \caption{Metrics values of the selected pairs of candidates for real data. The coordinates of the pair $\{a,b\}$ are $(\alpha(a,b),\beta(a,b))$, its size is $\gamma(a,b)$ and its color is determined by $\phi(a,b)$ with red corresponding to higher values. These data were gathered from 1000 profiles of 100 voters and 10 candidates.}
    \label{fig:res-real-small}
\end{figure}

In this set of experiments, we sampled rankings from the dataset over a subset of candidates in the case of real data or from a probability model in the case of synthetic data, in order to always have the same number of voters $n = 100$ and candidates $m=10$. We then look which pair of candidates would be selected by each rule. Then, we compare the values of their polarization metrics (partitioning ratio, discrepancy, discrepancy balance and group discrepancy imbalance).
This experiment is repeated over $1000$ random profiles. \Cref{fig:res-real-small} shows each selected pair~$\{a,b\}$ as a dot with coordinate $(\alpha(a,b), \beta(a,b))$; its size is
proportional to $\gamma(a,b)$ and the redder the dot the higher the value of $\phi(a,b)$. The first row contains these values for randomly selected pairs of candidates, in order to have an idea of the actual distribution of the pairs in these datasets. The figures for synthetic data and for more rules can be found in the Appendix \ref{appsec:expe-results}, as well as figures comparing the average values of the polarization metrics for the different rules.

The first observation, based on the first row of the figure in which we plot points associated to random pairs of candidates, is that the three datasets are very different in terms of conflict. The most conflictual seems to be the political one, with most pairs having both a high partitioning ratio $\alpha$ and discrepancy $\beta$. The sushi dataset is already less conflictual: Some pairs have higher $\alpha$, but their $\beta$ is quite low, and conversely. This is even worse for the ice skating dataset, as rankings rely a lot on the quality of participants' performance. %

The experiments confirm our theoretical
analysis of the formula based on $\alpha$ and
$\beta$, that we obtained assuming
the same distance between the candidates
for every vote. In particular, we clearly see that MaxNash
puts the most emphasis on $\beta$, while MaxSwap
puts the most
emphasis on having a clear division between $a$ and $b$ supporters (i.e.,
maximizing $\alpha$), but MaxSwap is also giving a lot of importance to group
discrepancy imbalance. On the contrary, 2-MaxPolar ignores
the discrepancy
balance, which explains why it might return pairs of candidates with smaller and
bluer dots.

\subsection{Results on a Conflictual Election} \label{sec:exp-conflictual}

\begin{table}[t]
    \centering
    \small
    \setlength{\tabcolsep}{2.5pt}
    \begin{tabular}{l|ccc|ccc}
             &   \multicolumn{3}{c}{2017} & \multicolumn{3}{c}{2022} \\ \midrule
       MaxSwap  & Far-left & \fight &Far-right   & Far-left &\fight& Far-right\\
       MaxNash  &  Socialist& \fight &Far-right &  Left& \fight& Far-right\\
       MaxSum &   Socialist&\fight& Far-right & Far-left  &\fight& Far-right\\
       2-MaxPolar &   Far-left&\fight& Far-right & Far-left&  \fight &Far-right\\
       \midrule 
       Borda & Left &\fight & Liberal  & Left &\fight& Green \\
       CC & Left&  \fight& Conservative& Green& \fight&  Far-right \\
    \end{tabular}
    \caption{Selected pairs for different rules.}
    \label{tab:res_french}
\end{table}

We now consider the political data, featuring the most conflictual
preferences, that was gathered by the {\em Voter Autrement initiative}
during the 2017 and 2022 presidential elections.%
\footnote{These datasets were collected by an online poll that asked voters to try alternative voting methods, including Instant Runoff Voting, thus giving us ranking data. We re-weighted the voters based on their vote at the official election such that the distribution is more faithful to the real distribution of opinion (and thus to the real polarization). Moreover, we only kept full rankings of preferences (voters had the possibility to rank only their top $k$ alternatives). This gives us $n = 5755$ (resp. $412$) voters and $m=11$ (resp. $12$) candidates for the 2017 (resp. 2022) dataset.}

\Cref{tab:res_french} summarizes the candidates selected by each rule. We also added the results for classic voting rules for comparison. We chose to put the political labelization of candidates instead of their names. While the non-conflictual rules return two {\em popular} candidates (e.g. the main candidate from the left and the main candidate from the right), conflictual rules tend to select at least one extreme candidate, if not two. Rules that put more emphasis on discrepancy like MaxNash would select more well-known candidates even if they are dividing the society less evenly, as voters have strong preferences on them. On the contrary, MaxSwap might select less well-known candidates, but who divide the society more evenly. MaxSum and 2-MaxPolar lie in between these two extremes.

\section{Conclusions and Future Work}
We proposed and analyzed rules that aim at selecting the most conflicting
candidates, and have shown that these rules fundamentally differ from the
standard ones. Together with the proposed metrics, our rules allow us to better
understand the structures causing conflict in the electorate.

A natural (and fairly non-trivial) follow up direction is to analyze extending
these rules and axioms to selecting more than just two candidates. Since our
rules provide a score for each pair of candidates, one polynomial-time
computable possibility would be to simply select such a committee, that
maximizes the smallest score among all pairs of candidates. Nonetheless, other
approaches evaluating the whole committee at the same time, rather than
pairs of candidates might turn out to be superior.

\section*{Acknowledgements}
This project has received funding from the European Research Council (ERC) under
the European Union’s Horizon 2020 research and innovation programme (grant
agreement No 101002854), and from the French government under management of
Agence Nationale de la Recherche as part of the "Investissements d'avenir"
program, reference ANR-19-P3IA-0001 (PRAIRIE 3IA Institute). The research
presented in this paper is supported in part from the funds assigned by Polish
Ministry of Science and Technology to AGH University.
\begin{center}
  \includegraphics[width=3cm]{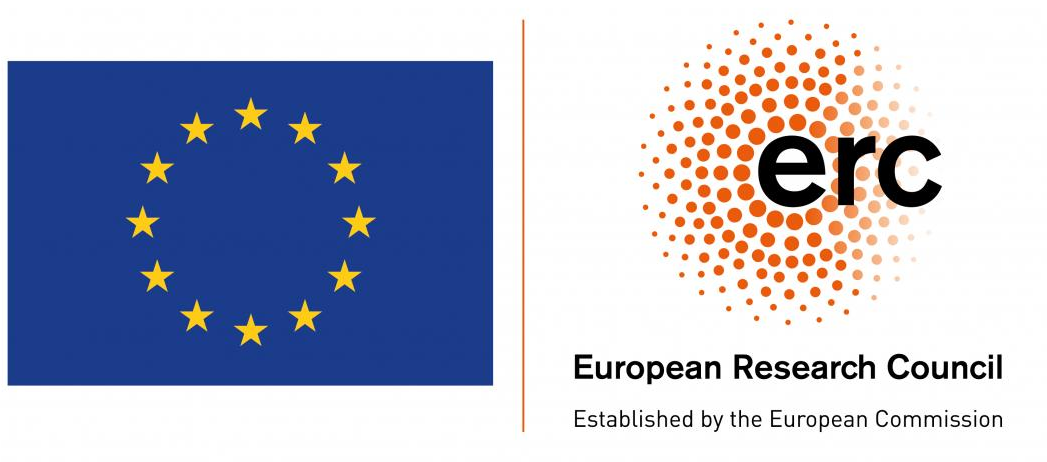}
\end{center}

\clearpage

\bibliographystyle{abbrvnat}
\bibliography{cvr}

\begin{thebibliography}{20}
\providecommand{\natexlab}[1]{#1}
\providecommand{\url}[1]{\texttt{#1}}
\expandafter\ifx\csname urlstyle\endcsname\relax
  \providecommand{\doi}[1]{doi: #1}\else
  \providecommand{\doi}{doi: \begingroup \urlstyle{rm}\Url}\fi

\bibitem[Alcalde-Unzu and Vorsatz(2013)]{alc-vor:j:cohesiveness}
J.~Alcalde-Unzu and M.~Vorsatz.
\newblock Measuring the cohesiveness of preferences: {A}n axiomatic analysis.
\newblock \emph{Social Choice and Welfare}, 41\penalty0 (4):\penalty0 965--988, 2013.

\bibitem[Alcalde-Unzu and Vorsatz(2016)]{alc-vor:j:cohesiveness2}
J.~Alcalde-Unzu and M.~Vorsatz.
\newblock Do we agree? {M}easuring the cohesiveness of preferences.
\newblock \emph{Theory and Decision}, 80\penalty0 (2):\penalty0 313--339, 2016.

\bibitem[Bouveret et~al.(2018)Bouveret, Blanch, Baujard, Durand, Igersheim, Lang, Laruelle, Laslier, Lebon, and Merlin]{voterautrement2018}
S.~Bouveret, R.~Blanch, A.~Baujard, F.~Durand, H.~Igersheim, J.~Lang, A.~Laruelle, J.-F. Laslier, I.~Lebon, and V.~Merlin.
\newblock Voter autrement 2017 - online experiment.
\newblock Dataset and companion article on Zenodo, July 2018.
\newblock URL \url{https://doi.org/10.5281/zenodo.1199545}.

\bibitem[Can et~al.(2017)Can, Ozkes, and Storcken]{can-ozk-sto:generalized-polarization}
B.~Can, A.~Ozkes, and T.~Storcken.
\newblock Generalized measures of polarization in preferences.
\newblock Technical report, HAL, 2017.

\bibitem[Chamberlin and Courant(1983)]{cha-cou:j:cc}
B.~Chamberlin and P.~Courant.
\newblock Representative deliberations and representative decisions: {Proportional} representation and the {B}orda rule.
\newblock \emph{American Political Science Review}, 77\penalty0 (3):\penalty0 718--733, 1983.

\bibitem[Colley et~al.(2023)Colley, Grandi, Hidalgo, Macedo, and Navarrete]{col-gra-hid-mac-nav:c:controlling-rank-aggregation}
R.~Colley, U.~Grandi, C.~Hidalgo, M.~Macedo, and C.~Navarrete.
\newblock Measuring and controlling divisiveness in rank aggregation.
\newblock In \emph{Proceedings of IJCAI-2023}, pages 2616--2623, 2023.

\bibitem[Debord(1992)]{deb:j:k-borda}
B.~Debord.
\newblock An axiomatic characterization of {B}orda's $k$-choice function.
\newblock \emph{Social Choice and Welfare}, 9\penalty0 (4):\penalty0 337--343, 1992.

\bibitem[Delemazure and Bouveret(2022)]{voterautrement2022}
T.~Delemazure and S.~Bouveret.
\newblock Voter autrement 2022 - online experiment (``{{U}}n autre vote'').
\newblock Dataset and companion article on Zenodo, April 2022.
\newblock URL \url{https://doi.org/10.5281/zenodo.10998451}.

\bibitem[Elkind et~al.(2017{\natexlab{a}})Elkind, Faliszewski, Laslier, Skowron, Slinko, and Talmon]{elk-fal-las-sko-sli-tal:c:multiwinner-voting}
E.~Elkind, P.~Faliszewski, J.~Laslier, P.~Skowron, A.~Slinko, and N.~Talmon.
\newblock What do multiwinner voting rules do? {{A}}n experiment over the two-dimensional euclidean domain.
\newblock In \emph{Proceedings of AAAI-2017}, pages 494--501, 2017{\natexlab{a}}.

\bibitem[Elkind et~al.(2017{\natexlab{b}})Elkind, Faliszewski, Skowron, and Slinko]{elk-fal-sko-sli:j:multiwinner-properties}
E.~Elkind, P.~Faliszewski, P.~Skowron, and A.~Slinko.
\newblock Properties of multiwinner voting rules.
\newblock \emph{Social Choice and Welfare}, 48\penalty0 (3):\penalty0 599--632, 2017{\natexlab{b}}.

\bibitem[Faliszewski et~al.(2017)Faliszewski, Skowron, Slinko, and Talmon]{fal-sko-sli-tal:b:multiwinner-voting}
P.~Faliszewski, P.~Skowron, A.~Slinko, and N.~Talmon.
\newblock Multiwinner voting: {A} new challenge for social choice theory.
\newblock In U.~Endriss, editor, \emph{Trends in Computational Social Choice}. AI Access Foundation, 2017.

\bibitem[Faliszewski et~al.(2019)Faliszewski, Skowron, Szufa, and Talmon]{fal-sko-szu-tal:c:stv-pav}
P.~Faliszewski, P.~Skowron, S.~Szufa, and N.~Talmon.
\newblock Proportional representation in elections: {STV} vs {PAV}.
\newblock In \emph{Proceedings of AAMAS-2011}, pages 1946--1948, 2019.

\bibitem[Faliszewski et~al.(2023)Faliszewski, Kaczmarczyk, Sornat, Szufa, and Wąs]{fal-kac-sor-szu-was:c:diversity-agreement-polarization}
P.~Faliszewski, A.~Kaczmarczyk, K.~Sornat, S.~Szufa, and T.~Wąs.
\newblock Diversity, agreement, and polarization in elections.
\newblock In \emph{Proceedings of IJCAI-2023}, pages 2684--2692, 2023.

\bibitem[Hashemi and Endriss(2014)]{has-end:c:diversity-indices}
V.~Hashemi and U.~Endriss.
\newblock Measuring diversity of preferences in a group.
\newblock In \emph{Proceedings of ECAI-2014}, pages 423--428, 2014.

\bibitem[Kamishima(2003)]{sushis}
T.~Kamishima.
\newblock Nantonac collaborative filtering: Recommendation based on order responses.
\newblock In \emph{Proceedings of KDD-03}, pages 583--588, 2003.

\bibitem[Mattei and Walsh(2013)]{preflib}
N.~Mattei and T.~Walsh.
\newblock Preflib: A library of preference data \textsc{http://preflib.org}.
\newblock In \emph{Proceedings of the 3rd International Conference on Algorithmic Decision Theory (ADT 2013)}, Lecture Notes in Artificial Intelligence. Springer, 2013.

\bibitem[Monroe(1995)]{mon:j:fully-proportional-representation}
B.~L. Monroe.
\newblock Fully proportional representation.
\newblock \emph{The American Political Science Review}, 89\penalty0 (4):\penalty0 925--940, 1995.

\bibitem[Peters et~al.(2021)Peters, Pierczy\'{n}ski, and Skowron]{per-pie-sko:c:proportional-pb}
D.~Peters, G.~Pierczy\'{n}ski, and P.~Skowron.
\newblock Proportional participatory budgeting with additive utilities.
\newblock In \emph{Proceedings of NeurIPS-2021}, pages 12726--12737, 2021.

\bibitem[Skowron et~al.(2017)Skowron, Lackner, Brill, Peters, and Elkind]{sko-lac-bri-pet-elk:c:proportional-rankings}
P.~Skowron, M.~Lackner, M.~Brill, D.~Peters, and E.~Elkind.
\newblock Proportional rankings.
\newblock In \emph{Proceedings of IJCAI-2017}, pages 409--415, 2017.

\bibitem[Skowron et~al.(2019)Skowron, Faliszewski, and Slinko]{sko-fal-sli:j:characterization-csr}
P.~Skowron, P.~Faliszewski, and A.~Slinko.
\newblock Axiomatic characterization of committee scoring rules.
\newblock \emph{Journal of Economic Theory}, 180:\penalty0 244--273, 2019.

\end{thebibliography}

\clearpage
\begin{center}
    \huge \textbf{Appendix}\bigskip
\end{center}

\appendix

\section{Missing Proofs}\label{apdx:proofs}

\appendixProofs

\subsection{Discrepancy}\label{apdx:discrepancy}
We prove here that $\beta_{\max} \in (\nicefrac{1}{3},1]$.
    For one voter, the sum of all $v(ab)$ for each pair of candidates is equal to 
    \begin{align*}
        \sum_{a,b \in C} |v(ab)| &= \sum_{i=1}^{m-1} \sum_{j=1}^i j \\
        &= \sum_{i=1}^{m-1} \frac{i(i+1)}{2} \\
        &= \frac{(m-1)(m)(2m-1)}{12} + \frac{m}{m-1}{4} \\
        &= \frac{(m-1)m(2m-1+3)}{12} \\
        &= \frac{(m-1)m(m+1)}{6}
    \end{align*}
    In the full profile, the sum of all $|v(ab)|$ is then equal to $\frac{n(m-1)m(m+1)}{6}$. Thus, the average sum of $|v(ab)|$ of a pair is equal to  
    \begin{align*}
       \frac{ \sum_{v\in V} \sum_{a,b \in C} |v(ab)|}{\sum_{a,b \in C} 1} &= \frac{\frac{n(m-1)m(m+1)}{6}}{\frac{m(m-1)}{2}} \\ 
       &= \frac{n(m+1)}{3}
    \end{align*}
    By pigeonhole principle, there exists at least one pair $\{a,b\}$ which have a sum higher than this average. Moreover, this average is achieved in the impartial culture, in which all pairs have the same score. Finally, this sum corresponds to the following value for $\beta$:
    \begin{align*}
       \beta(a,b) &= \frac{1}{n\cdot(m-1)}\sum_{v} |v(ab)|\\
       &\ge  \frac{1}{n\cdot(m-1)} \frac{n(m+1)}{3} \\
       &\ge \frac{1}{3}  \frac{m+1}{m-1}\\
       &\ge \frac{1}{3}
    \end{align*}
\paragraph{Characteristic Elections.}\label{apdx:compass}
Observe that the Antagonism (AN) profile have $\alpha_{\textrm{max}} = \beta_{\textrm{max}} = 1$. Then, in the Identity (ID) profile, $\alpha_{\textrm{max}}= 0$ but $\beta_{\textrm{max}} = 1$. Finally, in the Uniformity (UN) profile, we have $\alpha_{\textrm{max}} = 1$ but $\beta_{\textrm{max}} = \frac{m+1}{3(m-1)}$.

\section{Experiments}

In this section, we present in more depth the experiments conducted on our
conflicting rules. Note that the order of the subsections have been changed. In
\Cref{appsec:expe-mallows}, we present the experiments on Mallows models, in
\Cref{appsec:expe-real}, we present a more detailed analysis of the real
datasets. Finally, in \Cref{appsec:expe-results} we provide the results on all models.

\subsection{Mallows models} \label{appsec:expe-mallows}

In this section we present some experiments on Mallows models, and how conflicting the election  generated with these models are, depending on the used parameters.

The Mallows model takes two parameters: a \emph{central ranking} $\sigma$ and
parameter $\psi$. The idea is the following: $\sigma$ is the average ranking of
the voters, but they all deviate a bit from it. The higher the $\psi$, the
higher the deviation. More formally, if we denote by~$\text{KT}(\sigma_1,\sigma_2)$
the Kendall-tau distance between two rankings $\sigma_1,\sigma_2$ (i.e.,\ the
number of swaps needed to go from $\sigma_1$ to $\sigma_2$), then the
probability to generate the ranking $\sigma'$ from the Mallows model with
parameters $(\sigma,\psi)$ is $\psi^{\text{KT}(\sigma,\sigma')}/C$ where $C$ is
the normalization constant.

If $\psi = 0$, the profile is the Identity and everyone has the same ranking
$\sigma$; if $\psi = 1$, we obtain impartial culture and each ranking has the
same probability to be generated. So a Mallows instance with $0 < \psi < 1$
would be something that is between the Identity and Uniformity, and it will
probably not be antagonizing, as most people would have a similar ranking. 

To have more antagonistic profiles, we can sample rankings from \emph{two}
mallows models, sharing the same parameter $\psi$ but with different central
rankings $\sigma_i$ with $i \in \{1,2\}$. Then, we sample each ranking in the
profile randomly from one Mallows model with probability $1/2$. We expect those
2-Mallows models to have a more antagonistic structure.

To confirm this, we investigate 6 different models. We took $\psi \in \{0.1,
0.3, 0.6\}$ and generated random profiles from 1-Mallows and 2-Mallows models. Then, we look at the partitioning ratio $\alpha$, the discrepancy $\beta$, the discrepancy balance $\gamma$, and the group discrepancy imbalance $\phi$. \Cref{fig:mallows} shows these metrics for the different models. For $\psi = 0.1$, the voters preferences do not deviate much from the average ranking, thus there are some pairs with a high discrepancy (if they are far apart in the central ranking(s)). For 1-Mallows models, the partitioning ratio is very low for all pairs, but for 2-Mallows there is a clear division between the pairs that have higher partitioning ratio (if their relative order is different in the two central rankings $\sigma_1$ and $\sigma_2$) and the pairs that do not have it (i.e., if their relative order is the same). 

For $\psi = 0.3$ and the 1-Mallows model, the voters preferences deviate more from the central ranking, and we can already see that some pairs of candidates have low discrepancy $\beta$ and large partitioning ratio $\alpha$. However, we observe that the pairs with high $\alpha$ tend to have lower $\beta$ and vice-versa. It is less the case for the 2-Mallows model, but we also observe that the pairs have a higher partitioning ratio, i.e., the electorate is more divided. 

For $\psi = 0.6$, the preferences deviate so much from the central ranking(s) that we get very close to the impartial culture. This can be noticed by the fact that all pairs have similar metric values, and this in both models (1-Mallows and 2-Mallows). In particular, they have a high partitioning ratio (for impartial culture, the probability that $a \succ b$ and $b \succ a$ are almost the same), but an average value for discrepancy.

\begin{figure}[t]
    \centering
    \includegraphics[width=0.5\textwidth]{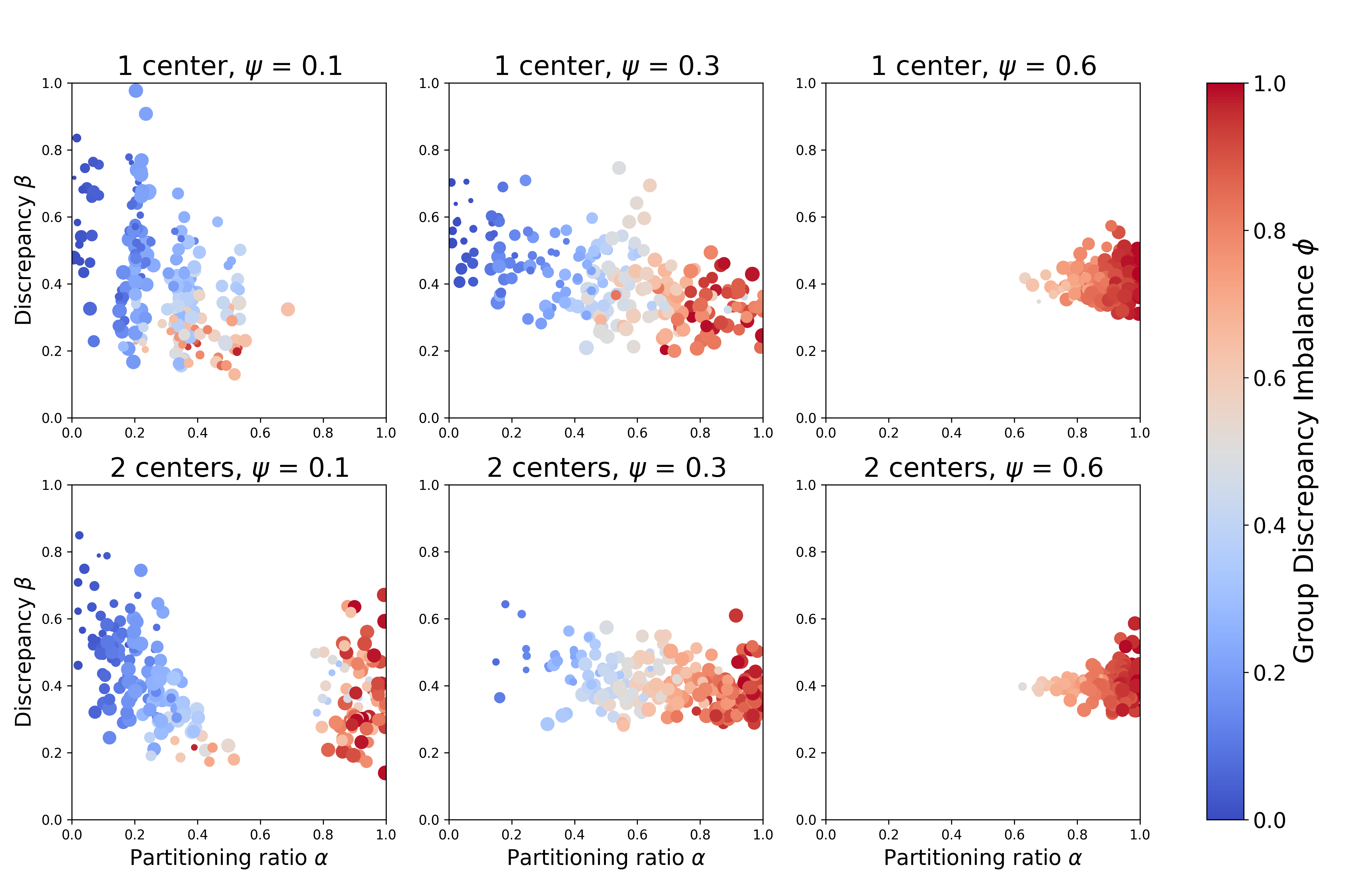}
    \caption{Metrics of the pairs of candidates with Mallows models. Each column corresponds to a parameter $\psi \in \{0.1,0.3,0.6\}$. The first row is for Mallows using 1 central ranking and the second row Mallows with 2 central rankings. The coordinates of the pair $\{a,b\}$ are $(\alpha(a,b), \beta(a,b)$, its size depends on $\gamma(a,b)$ and its color on $\phi(a,b)$. These data were gathered from 5 profiles for each model, each containing 1000 voters and 10 candidates.
    }
    \label{fig:mallows}
\end{figure}

\begin{figure}[t]
    \centering
    \begin{tikzpicture}

\definecolor{crimson2273120}{RGB}{227,31,20}
\definecolor{darkgray176}{RGB}{176,176,176}
\definecolor{greenyellow18622434}{RGB}{186,224,34}
\definecolor{lightgray204}{RGB}{204,204,204}
\definecolor{limegreen7522434}{RGB}{75,224,34}
\definecolor{slateblue7956214}{RGB}{79,56,214}

\begin{axis}[
legend cell align={left},
legend style={
  fill opacity=0.8,
  draw opacity=1,
  text opacity=1,
  at={(0.97,0.03)},
  anchor=south east,
  draw=lightgray204
},
tick align=outside,
tick pos=left,
x grid style={darkgray176},
xlabel={$\psi$},
xmin=0, xmax=1,
xtick style={color=black},
y grid style={darkgray176},
ylabel={Metrics values},
ymin=0, ymax=1,
ytick style={color=black},
height=5.1cm,width=8cm
]
\addplot [semithick, slateblue7956214, mark=*, mark size=2, mark options={solid}]
table {%
0 0
0.1 0.231133333333333
0.2 0.427130666666667
0.3 0.567047111111111
0.4 0.670825777777778
0.5 0.751882666666667
0.6 0.828278222222222
0.7 0.857457777777777
0.8 0.916997333333333
0.9 0.953928888888889
1 0.974720888888889
};
\addlegendentry{mean $\alpha$}
\addplot [semithick, crimson2273120, mark=*, mark size=2, mark options={solid}]
table {%
0 1
0.1 0.857686666666667
0.2 0.758597777777778
0.3 0.66944
0.4 0.610917777777778
0.5 0.542648888888889
0.6 0.514777777777778
0.7 0.466008888888889
0.8 0.448437777777778
0.9 0.431506666666666
1 0.424746666666667
};
\addlegendentry{max $\beta$}
\addplot [semithick, greenyellow18622434, mark=*, mark size=2, mark options={solid}]
table {%
0 0
0.1 0.605811779668681
0.2 0.670694691514365
0.3 0.720173700312688
0.4 0.773780279085839
0.5 0.819385908642286
0.6 0.87493191406333
0.7 0.897820808018615
0.8 0.936929010916018
0.9 0.958753507983522
1 0.970574874541157
};
\addlegendentry{mean $\gamma$}
\addplot [semithick, limegreen7522434, mark=*, mark size=2, mark options={solid}]
table {%
0 0
0.1 0.251851524622495
0.2 0.402016498721265
0.3 0.506709373019319
0.4 0.598911101961242
0.5 0.681196567958553
0.6 0.775218479445246
0.7 0.81178027598148
0.8 0.889608411209228
0.9 0.940637685349049
1 0.970207506812057
};
\addlegendentry{mean $\phi$}
\addplot [semithick, slateblue7956214, dashed, mark=*, mark size=2, mark options={solid}, forget plot]
table {%
0 0.496351111111111
0.1 0.569287111111111
0.2 0.646723555555555
0.3 0.729943111111111
0.4 0.782269333333333
0.5 0.820643555555556
0.6 0.871728888888889
0.7 0.913984888888889
0.8 0.934208
0.9 0.961023111111111
1 0.974987555555555
};
\addplot [semithick, crimson2273120, dashed, mark=*, mark size=2, mark options={solid}, forget plot]
table {%
0 0.817742222222222
0.1 0.705855555555556
0.2 0.635031111111111
0.3 0.591057777777778
0.4 0.530944444444444
0.5 0.50586
0.6 0.4833
0.7 0.459355555555556
0.8 0.43564
0.9 0.429195555555556
1 0.425044444444444
};
\addplot [semithick, greenyellow18622434, dashed, mark=*, mark size=2, mark options={solid}, forget plot]
table {%
0 0.274429453262787
0.1 0.655623188771598
0.2 0.725853348232508
0.3 0.793674013327474
0.4 0.827919329425819
0.5 0.865175345398584
0.6 0.902324380997061
0.7 0.933452992680248
0.8 0.947054927528143
0.9 0.96340550930736
1 0.970441527802214
};
\addplot [semithick, limegreen7522434, dashed, mark=*, mark size=2, mark options={solid}, forget plot]
table {%
0 0.335754392824319
0.1 0.458999559622525
0.2 0.562750200949937
0.3 0.658166004826033
0.4 0.719926643561405
0.5 0.766849512305874
0.6 0.830629449684167
0.7 0.884931753426254
0.8 0.912866835009222
0.9 0.950248592261619
1 0.970881199454743
};
\end{axis}

\end{tikzpicture}
    \caption{Mean or maximum values of different metrics for Mallows models with 1 (full line) or 2 (dotted line) central ranking(s). Each dot is averaged over 50 profiles.}
    \label{fig:mallows-phi}
\end{figure}
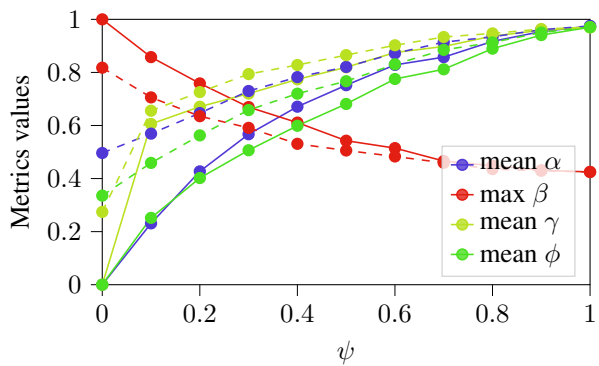

\Cref{fig:mallows-phi} shows the mean values (maximum value in the case of discrepancy) of some metrics for Mallows models, depending on the value of $\psi$ and the number of central rankings. First, we naturally observe more polarization when there are two central rankings, except for the maximum discrepancy. Second, it seems that the ``sweet spot'' of polarization, in which all $\beta$, $\alpha$, and $\gamma$ are high is between $\psi = 0.1$ and $\psi = 0.4$ (for 2-Mallows). Thus, in the next experiments, we will focus on these models.

\subsection{Analysis of the real dataset} \label{appsec:expe-real}

In this section, we describe in more details the real data that we used in the experiments presented in \Cref{sec:experiments}. The three datasets we considered are the following:
\begin{itemize}
		\item \textbf{French presidential elections:} These data come from surveys
			conducted in parallel to the actual 2017 and 2022 presidential elections
			in France. In this surveys, people were asked what would be their vote
			with alternative voting methods, in particular with instant runoff voting
			(IRV) \citep{voterautrement2018} for which voters had to give a rankings of
			the candidates. However, participants had the possibility to only rank a
			subset of the candidates. In these experiments, we removed all incomplete
			rankings and obtained $n=5,755$ rankings (resp. $412$) over $m=11$ (resp
			$12$) candidates for the 2017 (resp. 2022) dataset. Finally, we
			re-weighted the voters based on their vote at the official election to
			obtain a more faithful distribution of opinions. %
			distribution is more faithful to the actual distribution of opinion.
    \item \textbf{Sushi:} The sushi dataset, from Preflib~\citep{preflib}, contains preferences of $n=5,000$ voters over $m=10$ types of sushi.
    \item \textbf{Figure skating competition:} This dataset, also from Preflib, contains rankings of judges in figure skating competitions, after seeing the performances of the candidates. These rankings might contain ties that we decided to break arbitrarily. The dataset contains rankings for 49 competitions, each with between $m=10$ and $25$ candidates and between $n=8$ and $10$ judges.
\end{itemize}

We expect the political dataset to be the most conflictual, as people tends to have strong opinions on candidates at the presidential elections, and this kind of elections are based a lot around the ideological conflicts between the candidates. On the other hand, we expect the figure skating dataset to be the least conflictual, as judges will base their rankings on the performances of the candidates, and figure skating performance can be evaluated based on participants skills and advance of presented acrobatic figures.

In \Cref{fig:real-datasets}, we show the $(\alpha,\beta,\gamma,\phi)$-values for these 3 datasets. For the figure skating one, we took one random dataset among the 49 available, but they all have a similar structure regarding polarization. At first glance, the figure skating dataset structure looks very similar to what we obtained with 1-Mallows and $\psi = 0.1$. However, it differs as we still observe pairs of candidates with high partitioning ratio $\alpha$. However, these pairs with high partitioning ratio have low discrepancy. Conversely, there are pairs with very high discrepancy but low partitioning ratio. The first case corresponds to pairs of candidates that have very similar skills, but the jury is very divided between them. Thus, they are always next to each other in the rankings, but half of the voters think one is better and the other half that the other is better. The second case corresponds to pairs of candidates with a clear skills gap, such that all voters rank one very high and one very low, but they all rank them in the same order. 

The figure for the political dataset looks more like the ones of 2-Mallows, with pairs that have high partitioning ratio and others with low partitioning ratio but very few pairs with average partitioning ratio. Moreover, we see that some pairs have both high partitioning ratio and high discrepancy. Finally, the sushi dataset has a similar structure than 1-Mallows with $\phi \sim 0.5$, with not much pairs that are both conflicting globally (with high partitioning ratio) and locally (with high discrepancy).

\begin{figure}[t]
    \centering
    \includegraphics[width=0.5\textwidth]{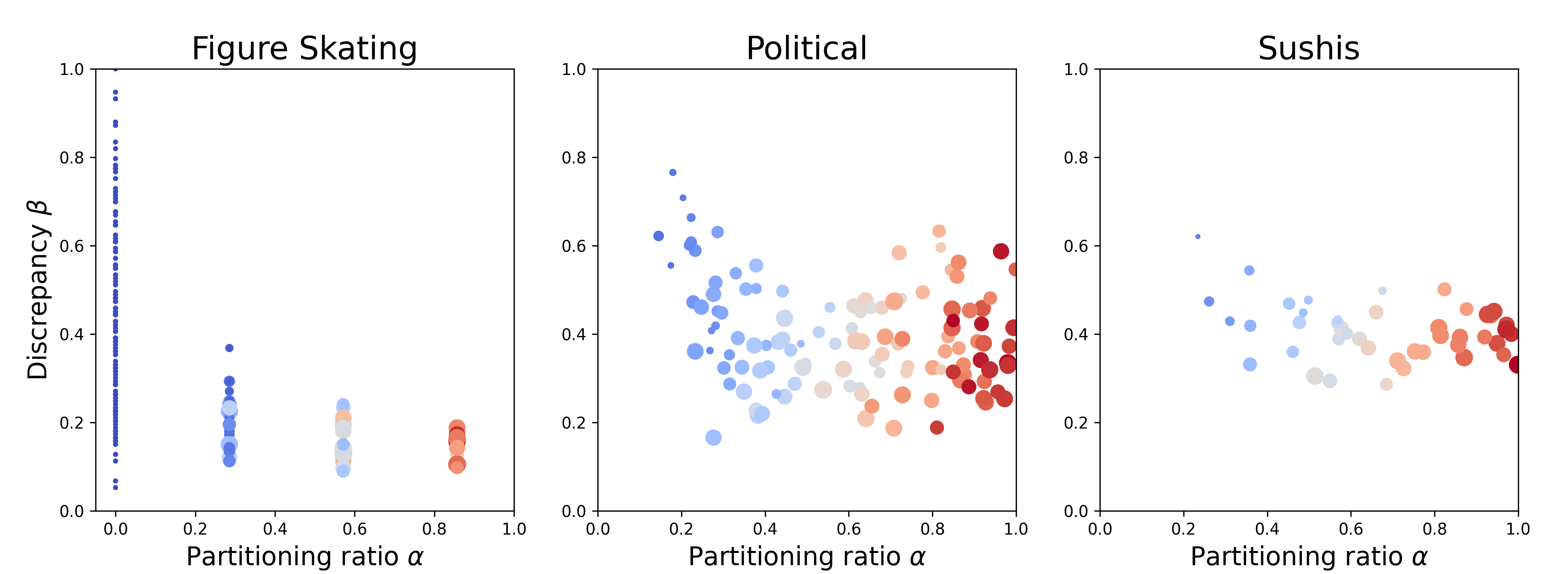}
    \caption{Metrics values of the pairs of candidates for different real datasets. The coordinates of the pair $\{a,b\}$ are $(\alpha(a,b),\beta(a,b))$, its size depends on $\gamma(a,b)$ and its color on $\phi(a,b)$.}
    \label{fig:real-datasets}
\end{figure}

\begin{table}[!t]
    \centering
    \small
    \setlength{\tabcolsep}{2.5pt}
    \begin{tabular}{l|ccc|ccc}
             &   \multicolumn{3}{c}{2017} & \multicolumn{3}{c}{2022} \\ \midrule
       MaxSwap  & Poutou & \fight & Le Pen   & Poutou &\fight& Zemmour\\
       MaxNash  &  Hamon & \fight & Le Pen &  Mélenchon& \fight& Le Pen\\
       MaxSum &   Hamon &\fight& Le Pen & Poutou &\fight& Le Pen\\
       2-MaxPolar &  Poutou &\fight& Le Pen & Poutou&  \fight &Le Pen\\
       \midrule 
       Borda & Mélenchon &\fight & Macron  & Mélenchon &\fight& Jadot \\
       CC & Mélenchon&  \fight& Fillon& Jadot& \fight&  Le Pen \\
    \end{tabular}
    \caption{Selected pairs for different rules.}
    \label{tab:res_french_2}
\end{table}

\begin{table}[!t]
    \centering
    \small
    \begin{tabular}{|l|l|c |c|}
    \toprule
         Name & Represented ideology  &  2017 & 2022 \\ \midrule
    Poutou & Far-left & $1.09\%$ & $0.77\%$ \\
      Mélenchon & Left & $19.58\%$ & $21.95\%$ \\
     Hamon &  Socialist & $6.36\%$ & N/A \\
    Jadot &  Green & N/A & $4.63\%$ \\
     Macron &  Liberal & $24.01\%$ & $27.85\%$ \\
      Fillon & Conservative & $20.01\%$ & N/A \\
       Le Pen &Nationalist& $21.30\%$ & $23.15\%$ \\
       Zemmour & Far-right& N/A & $7.07\%$\\
       \bottomrule
    \end{tabular}
    \caption{Score of some candidates at the presidential elections.}
    \label{tab:score_french}
\end{table}

For the political dataset, we complement the analysis from \Cref{sec:exp-conflictual}, with the names of the candidates selected by each rule in \Cref{tab:res_french_2}, and their score at the elections in \Cref{tab:score_french}. We can see that MaxSwap tends to select less popular candidates, but with higher $\alpha$, while MaxNash selects more popular ones.

\subsection{Complete results}\label{appsec:expe-results}

In this section, we present the results of our experiments for more rules (adding CC and Borda) and more datasets. All experiments were conducted on profiles consisting of $n=100$ voters and $m=10$ candidates, and averaged over $1000$ such profiles. For real data, we select a random subset of candidates and randomly draw rankings from the profile (based on weights if there are weights). We present here the results for (1) Metric models, (2) Mallows models and (3) Real data.

Figures from \ref{fig:res-metric-big} to \ref{fig:res-real} show the $(\alpha,\beta,\gamma,\phi)$-values of the selected pairs of candidates, as well as the ones of random pairs of candidates in the first row, to get an idea of the structure of the data. Figures \ref{fig:res-mean} and \ref{fig:res-zoom} show the average $(\alpha,\beta,\gamma,\phi)$-values of the selected pairs of candidates for each rule and model. The conclusions of these experiments are very similar to what we already discussed in \Cref{sec:experiments}.

\begin{figure*}
    \centering
    \includegraphics[width=0.58\textwidth]{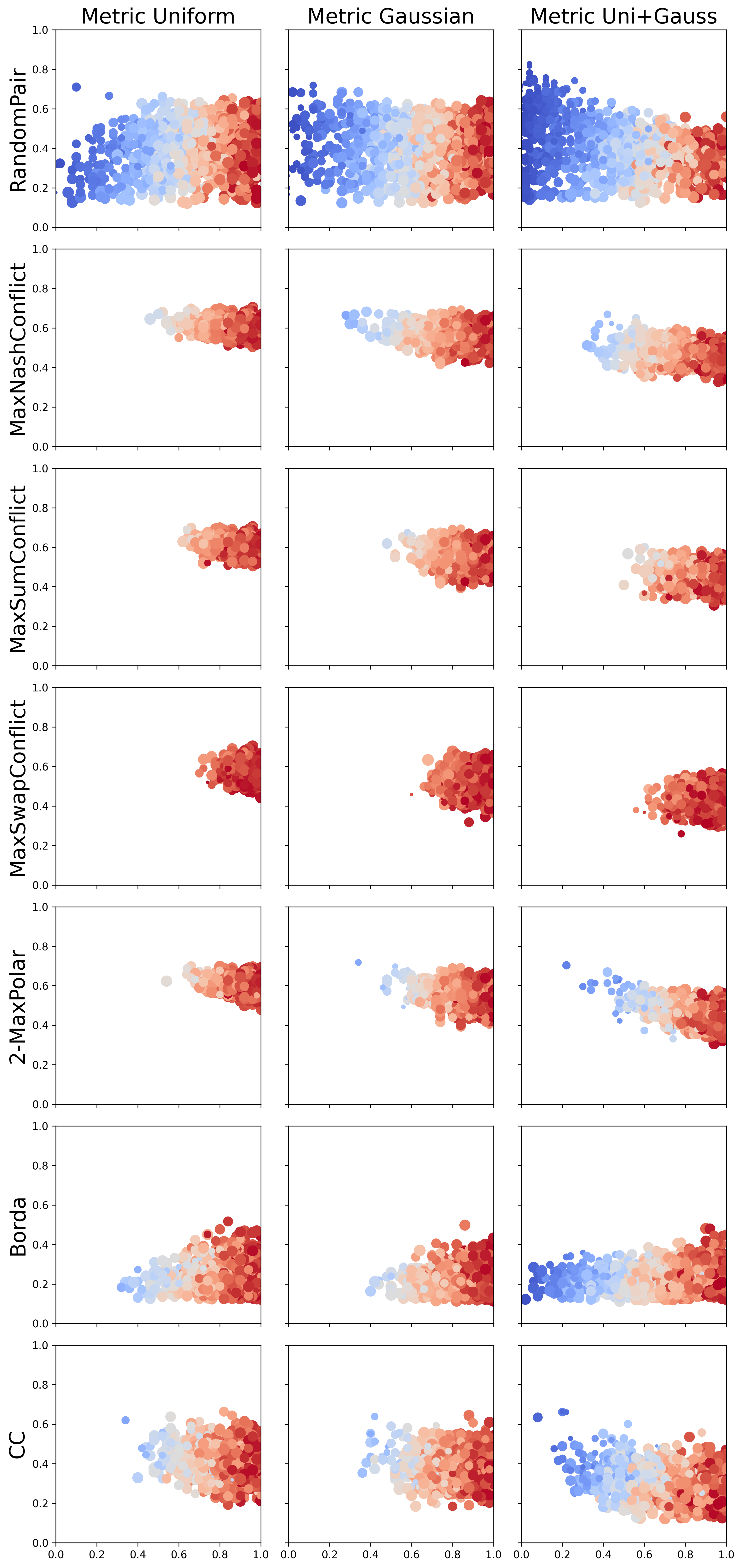}
    \caption{Results for distribution on 2D-Euclidean metric space, averaged over 1000 profiles of 100 voters and 10 candidates. ``Uni+Gauss'' means a distribution in which voters are drawn according to a Gaussian and candidates according to a Uniform distribution.}
    \label{fig:res-metric-big}
\end{figure*}

\begin{figure*}
    \centering
    \includegraphics[width=.6\textwidth]{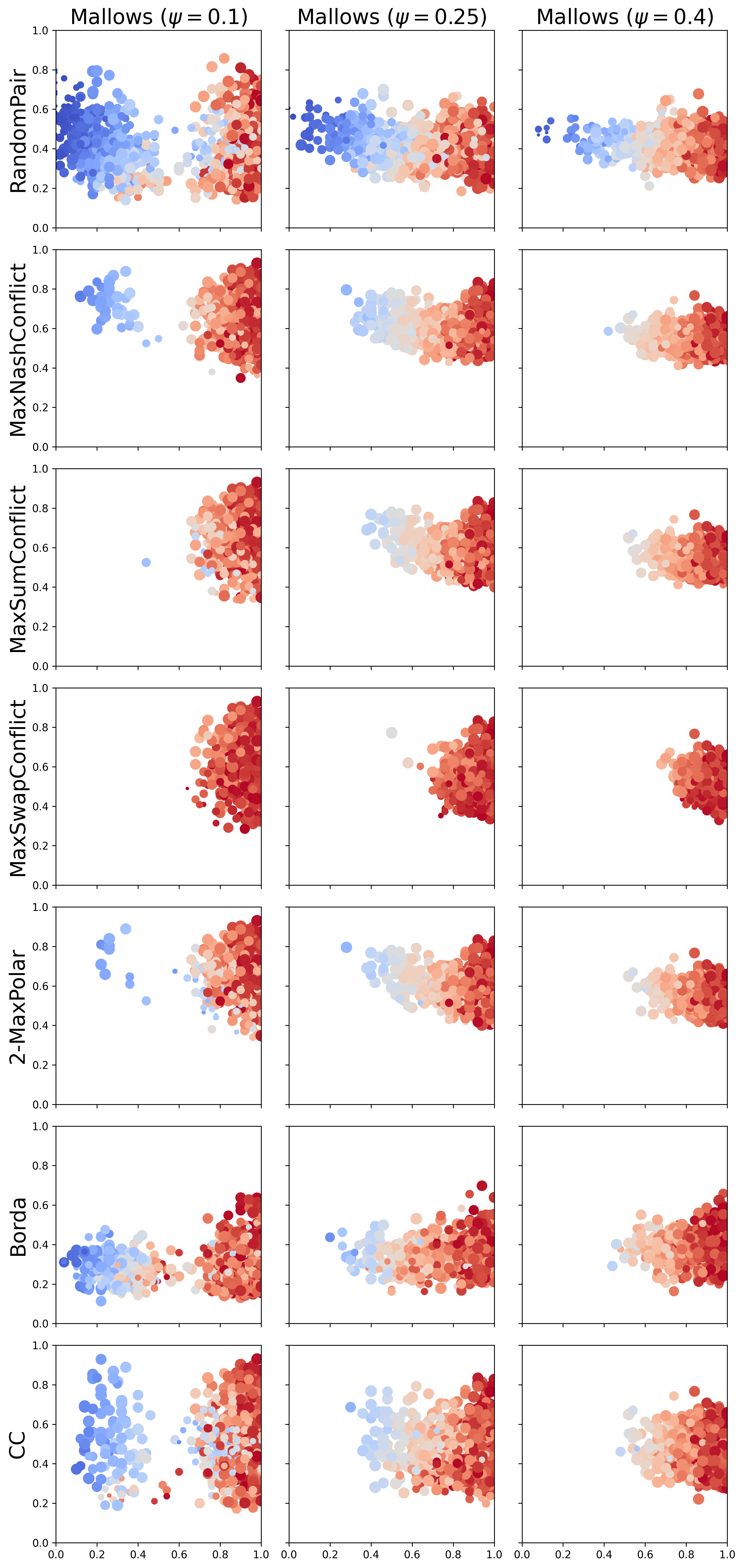}
    \caption{Results for 2-Mallows distributions, averaged over 1000 profiles of 100 voters and 10 candidates.}
    \label{fig:res-mallows}
\end{figure*}

\begin{figure*}
    \centering
    \includegraphics[width=.6\textwidth]{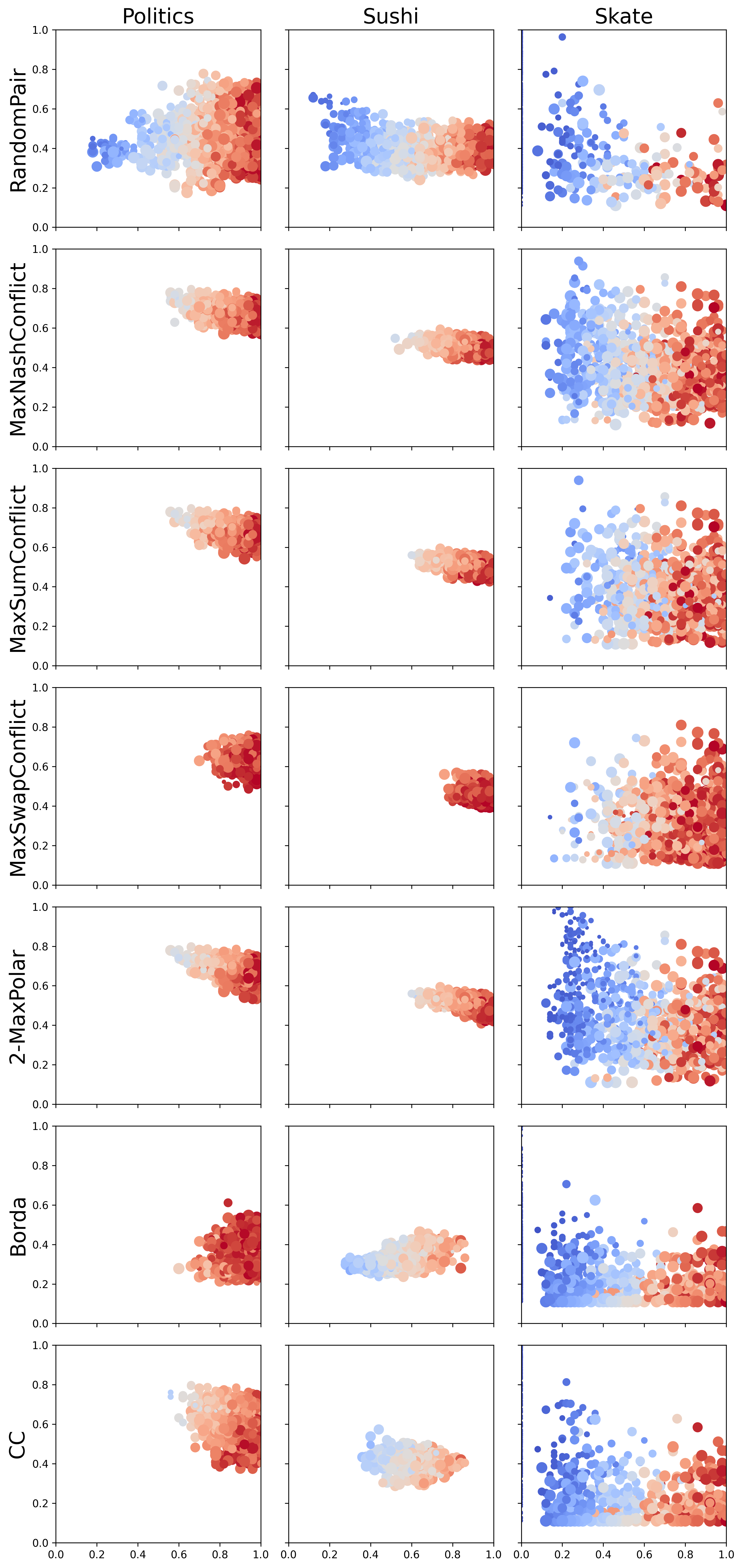}
    \caption{Results for real datasets, averaged over 1000 profiles of 100 voters and 10 candidates.}
    \label{fig:res-real}
\end{figure*}

\begin{figure*}
    \centering
    \includegraphics[width=.8\textwidth]{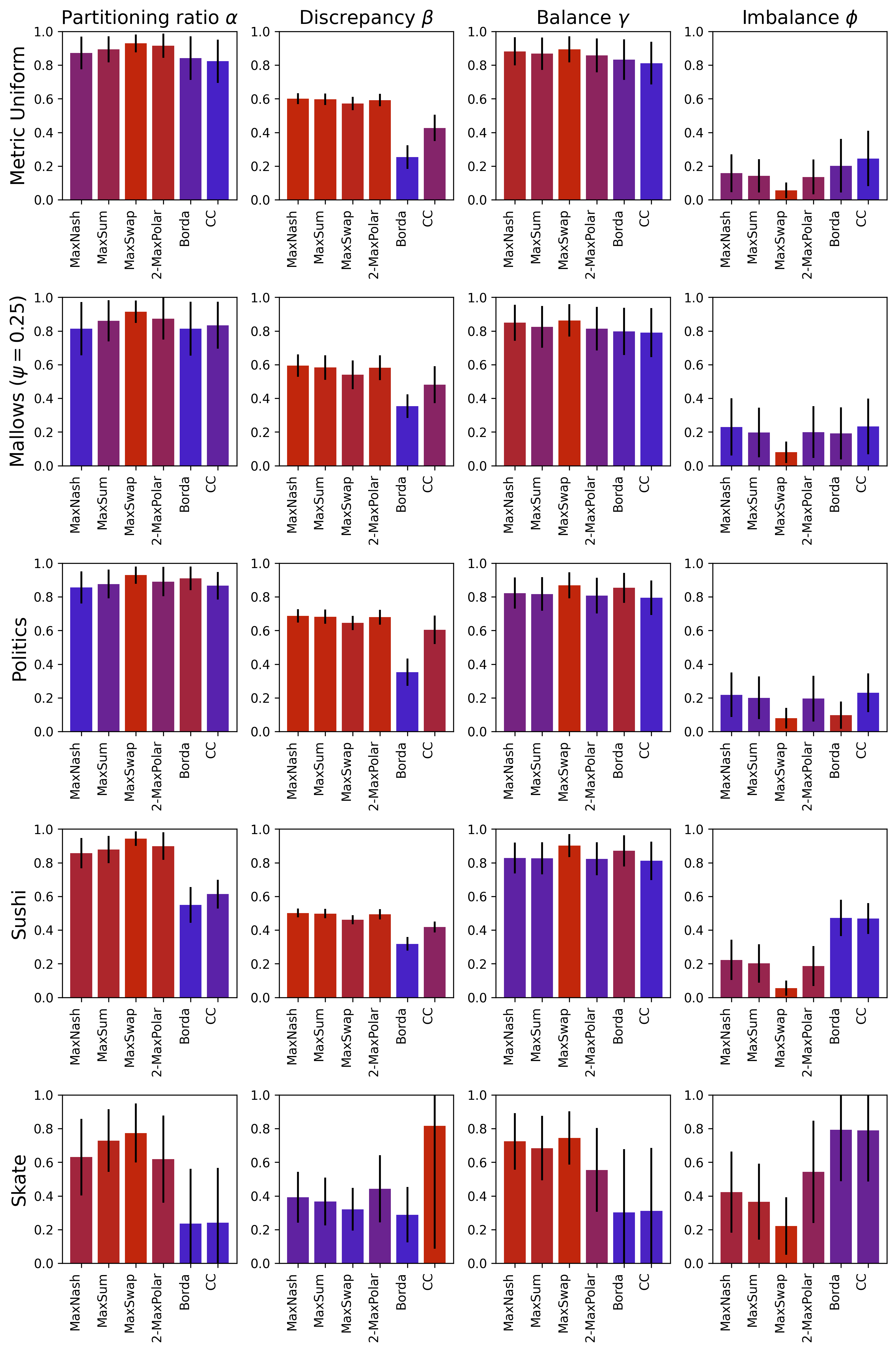}
    \caption{Average and deviation of metrics for different rules and different distributions, averaged over 1000 profiles of 100 voters and 10 candidates.}
    \label{fig:res-mean}
\end{figure*}

\begin{figure*}
    \centering
    \includegraphics[width=.8\textwidth]{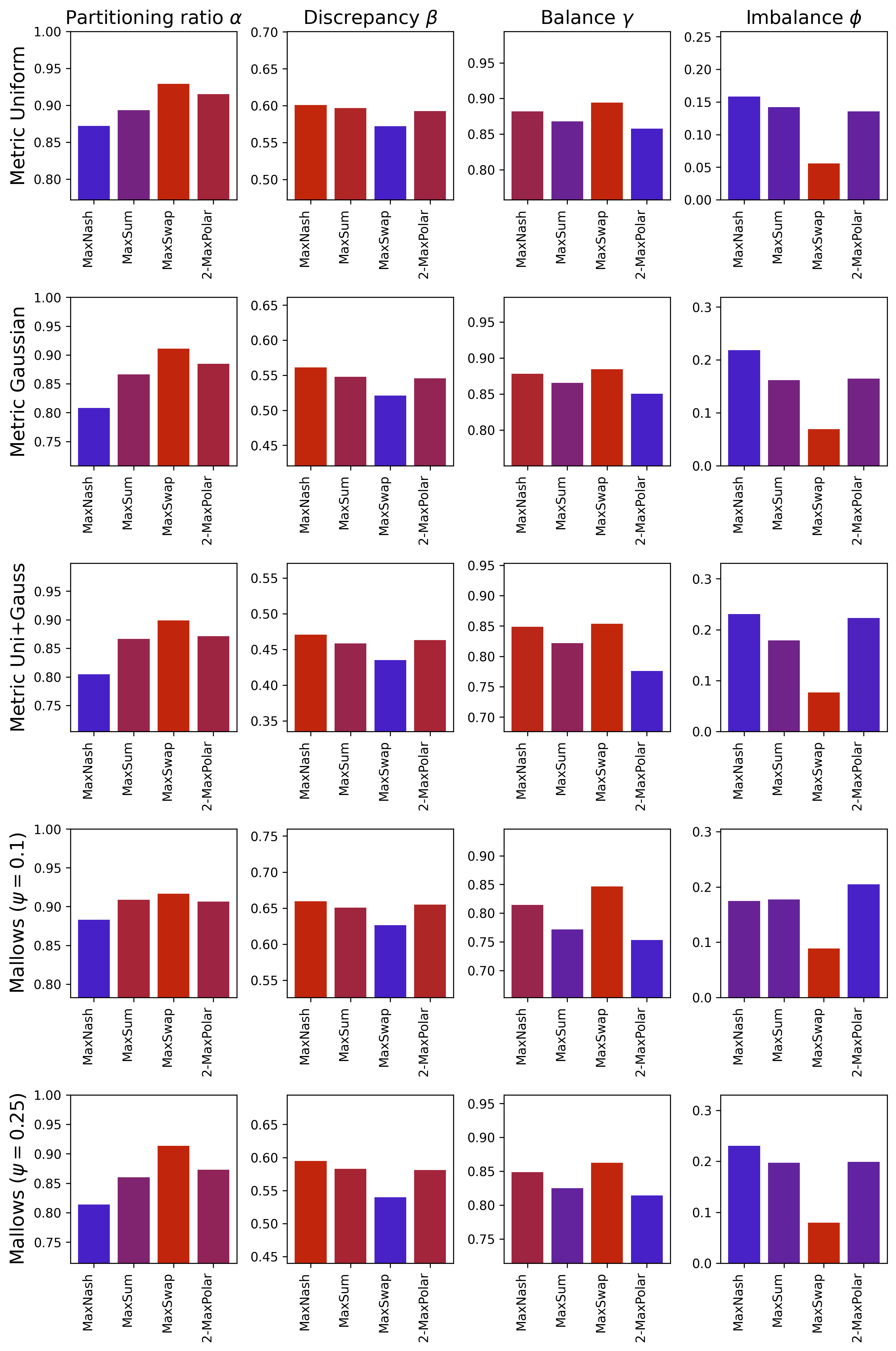}
    \caption{Average and deviation of metrics for different conflictual rules and different distributions, averaged over 1000 profiles of 100 voters and 10 candidates.}
    \label{fig:res-zoom}
\end{figure*}

\end{document}